\documentclass[11pt,a4paper]{article}
\usepackage[english]{babel}
\usepackage[T1]{fontenc}
\usepackage{amsmath, amssymb, amsfonts, amsthm}
\usepackage{indentfirst}
\usepackage{bbm,enumitem}
\usepackage{mathrsfs}
\usepackage{float}
\usepackage{color}
\usepackage{units}
\usepackage{graphicx}
\usepackage{caption}
\usepackage{subfig}
\usepackage{setspace}
\usepackage{esint}
\usepackage{multirow}
\usepackage[authoryear]{natbib}
\usepackage{hyperref}
\usepackage[left=2.25cm,right=2.25cm,top=2cm,bottom=3cm]{geometry}

\usepackage{chngcntr} 
\usepackage{apptools}
\usepackage{pgfplots}
\usetikzlibrary{calc,patterns}
\AtAppendix{\counterwithin{counter}{section}}

\setlist{nolistsep,noitemsep}

\newcounter{definitions}
\newcounter{examples}
\newcounter{propositions}
\newcounter{corollaries}

\newtheorem{proposition}[propositions]{Proposition}
\newtheorem{definition}[definitions]{Definition}

\newtheorem{lemma}{Lemma}
\newtheorem{corollary}[corollaries]{Corollary}
\newtheorem{theorem}{Theorem}
\newtheorem{example}[examples]{Example}

\begin{document}
	\title{Only Time Will Tell: Credible Dynamic Signaling\footnote{This paper is based on chapter 3 of the author's Ph.D. thesis. The author thanks Nemanja Anti{\'c}, Eddie Dekel, Jeffrey Ely, Yingni Guo, Nicolas Inostroza, Johan Lagerl{\"o}f, Alexey Makarin, Wojciech Olszewski, Marco Schwarz, Ludvig Sinander, Peter Norman S{\o}rensen, Bruno Strulovici
    and 
    seminar participants at Northwestern University and University of Copenhagen
    for valuable feedback and helpful comments.}}
	
	\author{Egor Starkov\footnote{Department of Economics, University of Copenhagen, {\O}ster Farimagsgade 5, bygning 26, 1353 K{\o}benhavn K, Denmark; e-mail: \href{mailto:egor@starkov.email}{egor@starkov.email}.}}
	
	\maketitle

\begin{abstract}
	This paper characterizes informational outcomes in a model of dynamic signaling with vanishing commitment power. It shows that contrary to popular belief, informative equilibria with payoff-relevant signaling can exist without requiring unreasonable off-path beliefs. The paper provides a sharp characterization of possible separating equilibria: all signaling must take place through attrition, when the weakest type mixes between revealing own type and pooling with the stronger types. The framework explored in the paper is general, imposing only minimal assumptions on payoff monotonicity and single-crossing. Applications to bargaining, monopoly price signaling, and labor market signaling are developed to demonstrate the results in specific contexts.
   	
	
\end{abstract}

\section{Introduction} \label{sec:intro}

An antitrust case was recently brought against iPhone producer Apple by Epic Games, which tries to determine ``whether the market for in-app purchases within the App Store is unfairly monopolistic, and whether iOS itself is a monopoly that should be opened up to third-party stores and side-loaded apps'' (\cite{verge_20210522}).
This can be seen as the culmination of long-lasting rumblings among the iPhone app developers. For example, as reported by the New York Times a year earlier:
\begin{quotation}
	Many companies and app developers complain that Apple forces them to pay its commission to be included in the App Store, which is crucial to reaching the roughly 900 million people with iPhones. ... ``If you’re not in the App Store today, you’re not online. Your business cannot function,'' said Andy Yen, the chief executive of ProtonMail. ... ``If you want to pass through their gates, they’re going to charge you 30 percent of your revenue.'' (\citet{nytimes_20200728})
\end{quotation} 

One of Apple's main defensive points in regards to its App Store monopoly, both within the scope of this case and more broadly, is that it never raised the commission it charges the app developers, even after it had allegedly acquired substantial monopoly power in the smartphone market.\footnote{
	``In the more than a decade since the App Store debuted, [Apple has] never raised the commission or added a single fee. In fact we have reduced them for subscriptions and exempted additional categories of apps.'' (Written testimony of Apple CEO Tim Cook, available in \citet{cook_testimony}). This testimony was referred to during the Epic v. Apple hearings (\citet{robertson_tweet}). In particular, during the closing remarks, ``Apple countered by claiming that its cut of app and in-app purchases ... didn't go up at all \emph{when Epic Games claimed Apple became a monopoly in 2010}'' (\citet{appleinsider_20210524}, emphasis added).}
Apple argued that it has, in fact, reduced the commission for some developers as recently as 2020 via its Small Business Program (see \citet{verge_20201118} for more details).
The judge in the Epic v. Apple hearings was not convinced by this argument and has argued that this is not indicative of competition: ``The issue with the \$1 million Small Business Program, at least from what I’ve seen thus far: that really wasn’t the result of competition. That seemed to be a result of the pressure that you’re feeling from investigations, from lawsuits, not competition,'' (\citet{verge_20210521}) -- said the judge, referring to the increasing scrutiny that Apple and other tech giants have been facing from the regulators in recent years.\footnote{See \citet{verge_20190513,verge_20191115,verge_20200616,washpost_20200729} for news coverage of four precedents pertaining to Apple within less than two years.}

Does Apple's argument have merit? Can the prices that a firm sets serve as evidence of its competitive environment when it faces regulatory pressure? When the firm has market power, it can still price low in an attempt to signal to the regulator that it is pressured by competition -- but it is also facing a stronger temptation to exploit its market power and set high prices than a genuinely competitive firm. The classical signaling theory a la \citet{Spe73} suggests that both pooling and separating equilibria can arise, i.e., prices may or may not be informative of the firm's competitive environment.
However, the standard signaling models are static, meaning that for their results to apply, the firm in our story needs to be able to not only set the price today, but also commit to not change it in the future. Without such commitment no single pricing decision would seemingly have enough weight to be informative, as originally noted in a different context by \citet{AP87}: if the firm could set some price that conclusively proves it faces strong competition, it would do so for a short time, document this decision, and use it as proof in the future (after reverting back to monopolistic pricing).

This paper shows that the scope for informative signaling, while limited, does in fact exist in dynamic settings without commitment, contrary to the intuition above. In particular, it explores a general signaling model, in which a single long-lived sender is privately informed of his type and engages in a repeated interaction with a receiver, where the periods are vanishingly short. The receiver makes inferences about the sender's type from his action choices, and the sender's payoff is increasing in his reputation with the receiver. We show that under appropriate monotonicity and single-crossing conditions on the sender's payoff function, payoff-relevant signaling is possible in this setting via what is effectively a war of attrition, in which all sender types pool on the same action, with the lowest type mixing between pooling with the rest and separating to a myopically optimal action. Beyond such attrition, actions are as informative as cheap talk -- meaning that attrition is the only way in which payoff-relevant signaling can proceed in this class of models. The contribution of this paper is both in showing the existence of a wedge between signaling and cheap talk in the setting under consideration (i.e., that signaling \emph{is} possible), and characterizing the equilibrium outcomes.

In the context of the Apple example, this paper implies that while singular price drops or cuts can not be a convincing evidence of a lack of monopoly power in equilibrium, persistent pricing at a low level can be suggestive of it. A monopolist would in such an informative equilibrium mix between maintaining the pooling (competitive) price and raising it to a monopoly level, revealing itself. However, there exists no perfectly separating equilibrium, in which the firm charges different prices when it has monopoly power and when it is competing against other firms in the industry. As a consequence, it is not possible to completely rule out the possibility that the firm has market power based on pricing decisions alone.
All this is demonstrated by the applied model in Section \ref{sub:antitrust}. 

The main model explored in this paper is substantially more general than the story above and is intended to serve as a framework suitable for many applications. To illustrate this, Section \ref{sec:app} also develops applications to bargaining and labor market signaling. 
In the former, a seller who privately knows the value of the item bargains with a single potential buyer over the price of the sale. In the latter, a student privately informed of his ability repeatedly chooses whether to invest in education, which may signal this ability to firms on the job market. Both problems are classical applications of signaling theory.
We show that the results from the general model apply, so all informative equilibria in each of these settings must take the attrition form.

Going back to the general model, it is worth noting that even though the attrition structure is restrictive, it allows for nontrivial equilibrium multiplicity. In addition to various possible combinations of informative and uninformative periods, multiple attrition outcomes can be sustained in equilibrium in any given period, which differ with respect to the probability of separation of the lowest type. So while attrition is the unique form that informative equilibria can take, multiple informative equilibria may exist in which attrition proceeds at different speeds. The extent to which this dimension of multiplicity manifests in a given setting depends on the richness of the action set.

Finally, this paper provides a takeaway regarding modelling assumptions that would be valuable to applied theorists investigating whether the receiver perfectly learns the sender's type (equivalently, whether social learning occurs) in a given setting in the limit as $t \to \infty$. In particular, if one adopts the simplifying assumption of there being only two types, then they could plausibly arrive at a conclusion that asymptotic learning is perfect. Yet, as this paper implies, this conclusion would not extend to the setting with finitely many types: learning can only occur regarding the lowest type, but cannot distinguish any of the higher types. In turn, this impossibility result for finite types does not necessarily extend to the setting with an interval of types, where asymptotic learning is possible again (\citet{FS10} provide an example of such model and equilibrium in the bargaining context). The latter observation implies that a finite-type approximation of a continuous-time signaling model may produce misleading results.

The fact that informative equilibria of attrition form exist in dynamic signaling models has been observed in applied models before. In particular, similar equilibria in specific settings have been obtained by \citet{Vin90,DL06,DG12,LL13,DL16,Dil17,KK18} in the context of bargaining; \citet*{SZZ16} in corporate finance; \citet*{Vet97,AAC17,GK19,SS18a} in industrial organization/marketing; \citet{SS18b} and \citet{Vong21} in cheap talk games; \citet{GP12} in disclosure games; \citet*{DTE18} in Dynkin games; \citet{Pei20} in trust games. 
The contribution of this paper is in setting up a general model that nests many of the models above and in identifying the sufficient conditions that yield uniqueness of attrition as the only informative equilbrium outcome.

Our analysis relies on the restriction of off-equilibrium path beliefs to be ``reasonable''. In particular, we adopt the assumption of non-increasing belief supports or, as labeled by \citet{OR90}, \emph{NDOC} (``Never Dissuaded Once Convinced'') assumption. As the name suggests, it implies that once the receiver has ruled out some type of the sender as impossible, the receiver stands by this belief and never again assigns positive probability to that type, including off the equilibrium path.
\citet{Kaya09} and \citet{Rod12a,Rod12b} have shown that in the absence of NDOC full instantaneous separation is possible in dynamic settings, since the sender's behavior can be disciplined by strong reputational threats in case of deviations. While the approach can be justified when the sender's type may change over time and hence needs constant re-verification, in other settings it is susceptible to a critique of using unreasonable off-path threats to sustain an equilibrium -- a practice typically reproved in the literature on equilibrium refinements for static signaling games, as well as equilibrium concepts for dynamic games.\footnote{C.f. \citet{BS87} and \citet{CK87} for signaling and Chapter 4 in \citet{Mye97} for extensive-form games respectively.} 

This paper belongs to the literature on signaling models, which developed from a seminal contribution by \citet{Spe73}. See \citet{Ril01} for an excellent survey of the early literature on static signaling models. \citet{AP87} were among the first to recognize that if signaling is viewed as a dynamic process, and the sender cannot commit (contractually or otherwise) to future actions, then perfectly separating outcomes may not be sustainable in equilibrium, as in our salesman example above. \citet{BP93} proposed a similar point in a contracting model, in which the informed party can propose to renegotiate after a contract is signed and before it is executed.
The literature has responded to this conceptual challenge by searching for aspects of such dynamic interactions which would neutralize this impossibility and restore perfectly separating equilibria. For example, \citet{Wei83} considers a model in which the sender derives explicit utility from signaling. \citet{NVD90} and \citet{Swi99} argue that perfect separation can be sustained via tacit collusion on the receivers' side.\footnote{
	In their setting, a worker is signaling ability via years of education, and firms can at any point offer the worker a job. In a perfectly separating equilibrium, able workers choose to acquire college education, while the less able workers enter the job market right after high school. As soon as a worker enters college, they are revealed as able, hence firms have incentives to offer them a position with a high wage immediately, without waiting for the worker to finish their degree. In a tacit collusion equilibrium, as soon as one firm makes such an offer, other firms immediately start a bidding war for this worker, thereby eliminating any gains that could be accrued by the deviating firm by hiring this worker. \label{foot:NVD_Swi}}
\citet{Rod12a,Rod12b} obtains perfectly separating outcomes in a dynamic signaling model, in which the receiver's beliefs violate NDOC, which is justified by the possibility that the sender's type can change over time, and thus needs to be constantly reaffirmed. This paper shows instead that informative -- albeit not perfectly separating -- equilibria can exist in dynamic settings without any of the aforementioned features.
\citet{Dil17,Hei18}, and \citet{Whi19} consider dynamic signaling models in which the sender's actions are only imperfectly observed by the receiver. While equilibria with perfect separation in strategies are possible in such settings, the receiver can not learn the sender's type instantaneously. This paper shows that exogenous noise in observations is not necessary to generate such an equilibrium with gradual learning, and the noise can stem from the sender's strategy instead.

The remainder of this paper is organized as follows. Section \ref{sec:model} sets up the general model and introduces the two assumptions that serve as sufficient conditions for our results: payoff monotonicity and NDOC. We then proceed to analyze two versions of this model. The two-type version in Section \ref{sec:twotypes} can be seen as an illustrative example. The version with finitely many types, which requires an additional single-crossing assumption, is then explored in Section \ref{sec:finite}. Section \ref{sec:app} considers applications to price signaling, bargaining and labor market signaling, setting up the respective models and verifying that the required assumptions hold. Section \ref{sec:concl} concludes. The proofs of most statements are relegated to Appendix A. Appendix B constructs an informative equilibrium in the context of the price signaling application from Section \ref{sub:antitrust}.

\section{Model} \label{sec:model}

\subsection{Primitives}

We will be looking at a continuous limit of a discrete-time infinite-horizon game. Time is indexed by $t \in \mathcal{T} \equiv \{0, dt, 2dt, ... \}$; all the results apply to the limit as $dt \to 0$.
There are two players: a long-lived sender (agent) and a receiver. The agent has some persistent \emph{type} $\theta \in \Theta$, where $\Theta \subseteq \mathbb{R}$ is finite. Equivalently, $\theta$ can be the state of the world that the agent is privately informed of.

In every period $t$, a Stackelberg-type sequential game is played between the agent and the receiver.\footnote{The results of the paper apply (after some relabeling) to simultaneous-move games as well, but for illustrative simplicity we focus on Stackelberg games in the analysis.}
First, the agent chooses an \emph{action} $a_t \in A$ from a compact action set $A$. This action choice affects the realization of a public \emph{outcome} $x_t \in X$, the distribution of which at time $t$ depends on the agent's type $\theta$ and the past history $h_t$. We assume that outcomes never allow to perfectly identify $\theta$: the support of $x_t$ conditional on $(h_t,a_t,\theta)$ is independent of $\theta$.\footnote{One could call $x_t$ a public ``signal''; we avoid this phrasing so as to not create confusion with the process of signaling through actions. These outcomes are introduced to demonstrate that the results in this paper hold even if some noisy public information is being revealed, as long as it does not allow identifying individual types. This includes both exogenous public news and public signals produced as a result of the agent's actions. }
After an outcome has realized, the receiver selects an action $b_t \in B$ from a compact set $B$.
A time-$t$ \emph{history} is defined as $h_t \equiv ((a_0,x_0,b_0), \dots, (a_{t-dt},x_{t-dt},b_{t-dt}))$, i.e., as the record of past actions and outcomes up to, but not including, time $t \in \mathcal{T}$. Let $\mathcal{H}_t$ denote the set of all such time-$t$ histories, and $\mathcal{H} \equiv \cup_{t \in \mathcal{T}} \mathcal{H}_t$ be the set of all such histories.
For any two periods $s>t$, we write $h_s \succ h_t$ if history $h_s$ succeeds history $h_t$, i.e., $h_s$ and $h_t$ coincide on $(0,\dots,t-dt)$. Weak succession is denoted as $h_s \succeq h_t$ and means ``either $s>t$ and $h_s \succ h_t$, or $s=t$ and $h_s=h_t$''.

The receiver begins the game with a commonly known prior \emph{belief} $p_0 \in \varDelta(\Theta)$ about the agent's type. In every period $t$ the receiver observes the agent's latest action $a_t$ and the realized outcome $x_t$, and updates her belief $p_t$ upon this information. Specifically, we use $p(h_t,a_t,x_t)$ to denote the resulting posterior belief. Hereinafter, belief $p_t$ is also referred to as the agent's \emph{reputation}. Further, we use $p(h_t)$ to denote the belief conditioned on $h_t$ -- i.e., at the beginning of period $t$ -- as a shorthand for $p(h_{t-dt}, a_{t-dt}, x_{t-dt})$, the receiver's belief at her previous action node.
For any $p_t$, let $S(p_t) \subseteq \Theta$ be the support of belief $p_t$, i.e., the set of types to which $p_t$ assigns positive weight.
With abuse of notation, let $S(h_t) \equiv S(p(h_t))$.

At the end of every period $t$, the agent receives flow payoff $u(a_t,b_t,\theta)$, and the receiver obtains payoff $w(a_t,b_t,x_t,\theta)$. 
For simplicity, we assume that the agent's payoff $u$ does not depend on the realized outcome $x_t$ except through the effect it has on the receiver's action $b_t$.
Both functions are assumed to be upper semi-continuous in the respective player's action: $\lim_{a \to a_t} u(a,b_t,\theta) \leq u(a_t,b_t,\theta)$ and $\lim_{b \to b_t} w(a_t,b,x_t,\theta) \leq w(a_t,b_t,x_t,\theta)$ for all $a_t,b_t,x_t,\theta$.
Further, assume that function $u$ is bounded. 

To avoid having to deal with the repeated game effects and focus purely on the signaling concerns, we will be working under the assumption that the receiver is myopic and only maximizes the current flow payoff. This assumption can be justified on its own merit in some settings, e.g., when ``the receiver'' is a proxy for a competitive market of receivers or a sequence of short-lived players.\footnote{
	For example, in the context of the labor market signaling, a worker is continuously signaling his ability to a population of competitive firms, which, in an attempt to get the worker, bid wages up to the worker's expected productivity, thereby eliminating any strategic element in wage offers. Alternatively, in the context of price signaling, a firm uses its price (and price history) to signal its product's value to changing generations of short-lived consumers.}
This assumption is not strictly necessary, although it simplifies the exposition. Footnote \ref{foot:myopy} describes the extent to which the results can be applied in a model with a strategic receiver.

The \emph{receiver's strategy} is $\mathbf{b}$: $\mathcal{H} \times A \times X \to B$. Strategy $\mathbf{b}^*$ is \emph{optimal} for the receiver at a given history $h_t$ if the action it prescribes maximizes the receiver's current flow payoff given the agent's strategy and the receiver's belief at that history
\begin{equation}
	\label{eq:rec_opt}
	\mathbf{b}^*(h_t,a_t,x_t) \in \arg \max_b \mathbb{E}[ w(a_t,b,x_t,\theta) \mid p(h_t) ],
\end{equation}
where the expectation is taken over $\theta$. For simplicity, assume \emph{no sunspots}: for any $a,x$, and any two histories $h_t, h_s$, if $p(h_t)=p(h_s)$ then $\mathbf{b}^*(h_t,a,x) = \mathbf{b}^*(h_s,a,x)$.
Moving on to the agent, define a \emph{bliss} (myopically optimal) action set for the agent of type $\theta$ at history $h_t$ given receiver's strategy $\mathbf{b}$ as
$$A^*(h_t, \mathbf{b}, \theta) \equiv \arg \max_{a \in A} \left\{ \mathbb{E} [u(a, \mathbf{b}(h_t,a,x_t), \theta)] \right\},$$
where the expectation is taken over $x_t$.
Note that $\mathbf{b}^*$ exists and $A^*$ is non-empty due to the upper semi-continuity of the respective utility functions $w$ and $u$.
A \emph{pure strategy} for the agent of type $\theta$ is $\mathbf{a}_\theta: \mathcal{H} \to A$. Given some belief system $p$ and the receiver's strategy $\mathbf{b}$, let $U(\mathbf{a}_\theta | h_t, \mathbf{b}, \theta)$ denote the expected discounted continuation utility of type $\theta$ from following strategy $\mathbf{a}_\theta$ starting from $h_t \in \mathcal{H}$:
\begin{equation*}
	U(\mathbf{a}_\theta | h_t, \mathbf{b}, \theta) \equiv \mathbb{E} \left[ \sum_{s \in \mathcal{T}, s \geq t} e^{-r(s-t)} u \left(\mathbf{a}_\theta(h_s), \mathbf{b}_s, \theta \right) dt \mid h_t, \theta \right],
\end{equation*}
where $r$ is the agent's discount rate, and the expectation is taken over future outcomes $x_s$.
The agent is assumed to maximize his expected discounted sum of utilities.
Strategy $\mathbf{a}_\theta$ is \emph{optimal} for the agent of type $\theta$ given belief system $p$ and receiver's strategy $\mathbf{b}$ if it maximizes his continuation payoff at every history $h_t \in \mathcal{H}$, i.e., if
\begin{equation}
	\label{eq:util}
	U(\mathbf{a}_\theta | h_t, \mathbf{b}, \theta) = V (h_t,\mathbf{b}, \theta)  \equiv \max_{\mathbf{a}} \left\{ U(\mathbf{a} | h_t, \mathbf{b}, \theta) \right\},
\end{equation}
where $V (h_t,\mathbf{b},\theta)$ is hereinafter referred to as the value function.
With a slight abuse of notation we let 
\begin{align*} 
	V (a|h_t,\mathbf{b},\theta) \equiv \mathbb{E}_{x_t} \left[ u(a,\mathbf{b}(h_t,a,x_t),\theta)dt + e^{-rdt} V(h_{t+dt}, \mathbf{b}, \theta) \mid h_t, a, \mathbf{b}, \theta \right]
\end{align*} 
denote the highest expected continuation utility that type $\theta$ can achieve conditional on taking action $a$ at history $h_t$. The outer expectation is taken w.r.t. the realization of period-$t$ outcome $x_t$, which affects the receiver's belief $p_t$ and, thus, her action $b_t$ and the agent's contemporaneous utility $u(a_t,b_t,\theta)$. The $t+dt$-history is $h_{t+dt} = (h_t, (a_t, x_t,b_t))$.

Finally, a \emph{behavioral strategy} for the agent of type $\theta$ is $\alpha_\theta: \mathcal{H} \to \varDelta(A)$. By the Kuhn's Theorem (\citet{aumann_mixed_1964}), behavioral strategies are equivalent to mixed strategies in this setting. Let $\alpha_\theta (a|h_t)$ denote the probability with which action $a$ should be played by type $\theta$ after history $h_t$ according to strategy $\alpha_\theta(h_t)$.
A behavioral strategy $\alpha_\theta$ is then optimal for $\theta$ if there exists an equivalent mixed strategy (i.e., a probability distribution over pure strategies), such that all pure strategies in its support are optimal.

\subsection{Equilibrium Concept}

Introduced above is a dynamic game of incomplete information. The greatest common factor among the solution concepts used for this class of games (and requiring belief consistency) is Perfect Bayesian Equilibrium (PBE). In such an equilibrium, all players maximize their expected continuation payoffs given their beliefs about other players' actions and beliefs, and these beliefs must be consistent on path with the players' knowledge of the game. 

\begin{definition}
	A Perfect Bayesian Equilibrium is given by an agent's strategy profile $\alpha= \{\alpha_\theta\}_{\theta \in \Theta}$ with $\alpha_\theta: \mathcal{H} \to \varDelta(A)$, a receiver's strategy $\mathbf{b}$: $\mathcal{H} \times A \times X \to B$, and a belief system $p: \mathcal{H} \times A \times X \to \varDelta(\Theta)$ such that:
	\begin{enumerate}
		\item for all $\theta$: strategy profile $\alpha_\theta$ is optimal for the agent of type $\theta$;
		\item strategy $\mathbf{b}$ is optimal for the receiver at all histories $h_t \in \mathcal{H}$;
		\item belief $p$ is updated using Bayes' rule whenever possible.
	\end{enumerate}
\end{definition}

Our main results characterize signaling in all PBE that satisfy asumption (NDOC) as defined in the following subsection, hence they will also apply if one imposes additional restrictions or equilibrium refinements on top of PBE with (NDOC).
As mentioned previously, we explore equilibria for small but positive $dt$, and we are interested in the properties of these equilibria as $dt \to 0$.

\subsection{Assumptions} \label{sub:assum}

The two sections above define the primitives of the model but impose only very minimal restrictions on them. This section describes the two significant assumptions that will be imposed throughout and which are sufficient for the results in the two-type model: Monotonicity and Never Dissuaded Once Convinced. (The version of the model with more than two types requires a third assumption, Single Crossing, which is introduced and discussed separately, in Section \ref{sub:sc}.)
To introduce these assumptions, a few extra bits of notation would prove useful. Firstly, let $\delta_{\theta}$ denote the Dirac delta: given some $\theta \in \Theta$, $p(h_t) = \delta_{\theta}$ is equivalent to $S(h_t) = \{\theta\}$. 
Secondly, let $\tilde{u}_t$ denote the agent's induced flow payoff function given some fixed strategy $\mathbf{b}$ of the receiver: 
\begin{equation}
	\tilde{u}_t (a_t,p(h_t,a_t,x_t),\theta) \equiv u(a_t,\mathbf{b}(h_t,a_t,x_t),\theta).
	\label{eq:utilde}
\end{equation}
The assumptions can then be phrased as follows (with the discussion following afterwards):
\begin{trivlist}
	\item
	\begin{description}[align=right,labelwidth=2.5cm]
		\item[(MON)] Flow payoff function $\tilde{u}_t(a_t,p_t,\theta)$ is weakly increasing in $p_t$ (w.r.t. FOSD order) for all $t,h_t,a_t,\theta$ and all optimal $\mathbf{b}$.\footnote{Monotonicity w.r.t. FOSD order on $p_t$ is understood in the usual way: for any $p ' , p '' \in \varDelta(\Theta)$ such that $p '(\theta ') > p '' (\theta ')$, $p '(\theta '') < p '' (\theta '')$ for some $\theta ' > \theta ''$, and $p '(\theta) = p '' (\theta)$ for all $\theta \in \Theta \backslash \{\theta ', \theta '' \}$, it should be that $\tilde{u}_t(a_t,p',\theta) \geq \tilde{u}_t(a_t,p'',\theta)$. \label{foot:fosd}} 
		Further, if $p_t >_{FOSD} \delta_\theta$ then $\tilde{u}(a_t,p_t,\theta) > \tilde{u}(a_t,\delta_{\theta},\theta)$, and if $p_t <_{FOSD} \delta_\theta$ then $\tilde{u}(a_t,p_t,\theta) < \tilde{u}(a_t,\delta_{\theta},\theta)$.
		\item[(NDOC)] Process $p_t$ is progressively absolutely continuous. I.e., belief supports are non-increasing: for any $h_s \succ h_t$, $S(h_s) \subseteq S(h_t)$.
	\end{description}
\end{trivlist}

The first assumption, (MON), requires that the agent's flow payoff function is increasing in his reputation $p_t$. This captures the core idea of signaling models: the agent would like to signal that his type is high because that induces a favorable reaction from the receiver. For example, a firm with a reputation for quality product is more likely to sell more units, an able worker is more likely to be offered a job, and a strong bargainer is more likely to see the opponent conceding to a demanding offer. 
Monotonicity is primarily a restriction on the model primitives, namely the utility functions: given the receiver's preferences $w$, her optimal strategy $\mathbf{b}$ is unique up to indifference for any $a$ and $p(h_t)$. This makes $\tilde{u}(a,p,\theta)$ a well-defined function given some tie-breaking rule for the receiver. Hence given $w$, (MON) is a condition on the agent's utility function $u(a,b,\theta)$.\footnote{The assumption of receiver's myopy was introduced to render (MON) expressible in terms of model primitives. Any other set of assumptions on the players' payoffs and/or the receiver's behavior that yields (MON) would be sufficient for our results to hold. Developing such assumptions for the case of strategic receiver is not trivial, since the folk theorem dictates that any individually rational payoff for the sender can be sustained in equilibrium regardless of his reputation, which means (MON) does not hold without some additional assumptions on the receiver's behavior. \label{foot:myopy}} 
While the condition is phrased using weak monotonicity, strict preferences relative to degenerate reputation are required to guarantee the presence of signaling effects: any type must always be strictly willing to pool with the higher types and to separate away from the lower types. 

The second assumption, (NDOC), is the refinement of the equilibrium beliefs that drives our analysis. In particular, it says that if $p(\theta|h_t) = 0$ then $p(\theta|h_s) = 0$ for any pair of histories $h_s \succ h_t$ in $\mathcal{H}$. Note that this applies both on and off the equilibrium path. In other words, once the receiver is convinced that a given type of the agent is inconsistent with the evidence (the observed history), she can never be dissuaded from this conviction. This restriction appears reasonable, since over time, only more evidence is collected, but the existing evidence is never forgotten -- including the evidence that lead the receiver to rule out certain types of the agent at the time.\footnote{As argued by \citet{OR90}, p.97: ``[I]f we allow a player in a game of incomplete information to change his mind after he has been persuaded that he is playing with certainty against a given type, then why we do not do so in a game of complete information?''}

Alternatively, (NDOC) can be seen as a weak form of renegotiation-proofness in some settings, as suggested by \citet{EV03}. For example, in the context of labor market signaling, suppose that a firm offers a contract to a high school graduate, according to which it would hire the worker as soon as he obtains a college degree. Suppose further that in equilibrium such a contract is only accepted by able workers (whose cost of learning is low), while less able workers reject it in favor of getting a job immediately. Then if such a contract is accepted, the firm knows the worker is able, and it is in the best mutual interest of the firm and the worker to renegotiate the contract to start the job immediately, since the delay to obtain education is wasteful for both parties.

The (NDOC) assumption has been originally introduced by \citet{OR90}. It has been widely used in applied dynamic models with asymmetric information: one can find analogs of (NDOC), often labelled differently, in \citet*{Rub85,GP86,LeB92,Vet97,KWZ95,Sen00,EV03,FS05,Lai14,BZ16,GK19,SS18b,SS18a}.
(NDOC) is nonetheless a strong assumption and has been criticized as leading to possible equilibrium nonexistence (see \citet*{MTW87} and \citet{NVD90b}). This paper hence characterizes the equilibria conditional on existence, without making any existence claims. However, see Appendix B as well as references above for a number of examples of settings in which equilibria exist.

To simplify the analysis, we strengthen (NDOC) by rendering the receiver pessimistic off the equilibrium path -- her beliefs off path must put all weight on the lowest type among those she has not yet ruled out. This stronger condition is labeled as (NDOC-P) and is defined as follows:
\begin{trivlist}
	\item
	\begin{description}[align=right,labelwidth=2.5cm]
		\item[(NDOC-P)] The off-equilibrium-path beliefs are such that after any action $a$ that is not on equilibrium path at $h_t \in \mathcal{H}$: $p(h_t, a, x_t) = \delta_{\min S(h_{t}) }$ for any $x_t \in X$.\footnote{On-pathness is defined in the usual way; see Section \ref{sub:attrn} for a formal definition.}
	\end{description}
\end{trivlist}
Given (MON), this condition imposes the strongest possible punishment on the sender for any deviation among those punishments that satisfy (NDOC). Therefore, we argue that for any equilibrium that satisfies (NDOC), there exists an equivalent one that satisfies (NDOC-P), despite the latter being a stronger condition.
This claim is formalized by the following lemma, with the proof available in Appendix A.
\begin{lemma} \label{lem:cond}
	If (MON) holds then for any equilibrium that satisfies (NDOC), there exists a payoff-equivalent and on-path strategy-equivalent equilibrium that satisfies (NDOC-P).
\end{lemma}

\section{Two Types} \label{sec:twotypes}

This section explores the version of the model with only two types: $\Theta = \{L,H\}$. Here we show that signaling must take the form of attrition regardless of the sender's payoffs, as long as they are monotone in reputation $p_t$. The first part of Theorem \ref{thm:2types} states that \emph{perfect} separation cannot occur at any history in equilibrium: if a given action is on path for $\theta=H$ then it is also on path for $\theta=L$. 
This statement captures the idea of \citet{AP87} and \citet{NVD90}.
We also observe that there may effectively be only one such pooling action in any period, in the sense of all pooling actions must be payoff-equivalent for all types of the agent. This follows trivially from the fact that both types must be indifferent between playing any such action if there are more than one.

The insight that is novel (in the general setting) is that the converse to the first statement is not necessarily true: if $\alpha_L(a|h_t) > 0$ then $\alpha_H(a|h_t)$ may or may not be positive. In other words, there may exist actions which perfectly identify the low type, even if there do not exist any that identify the high type. It is immediate that the low type must be mixing for this to be possible. All this is summarized by the second part of the theorem. 
The statement does not claim existence of any such separating actions, since they, as previously mentioned, need not exist in any given case -- though Appendix B presents an example in which such an informative equilibrium exists.

\begin{theorem} \label{thm:2types}
	Suppose $\Theta = \{L,H\}$, (MON) holds, and $dt \to 0$. In any equilibrium $(\alpha,\mathbf{b},p)$ such that (NDOC-P) holds, at any $h_t \in \mathcal{H}$ with $S(h_t) = \{L,H\}$, and for any $a' \in A$: 
	\begin{enumerate}
		\item if $\alpha_H(a'|h_t) > 0$ then $\alpha_L(a'|h_t) > 0$. Further, all such $a'$ are payoff-equivalent in the sense that $V (a'|h_t, \mathbf{b}, \theta)$ is the same across such $a'$ for both types $\theta$.
		\item if $\alpha_H(a'|h_t) = 0$ and $\alpha_L(a'|h_t) > 0$ then $a' \in A^*(h_t, \mathbf{b}, L)$ and $V (a'|h_t, \mathbf{b}, L) =  V (a''|h_t, \mathbf{b}, L)$ for any $a''$ such that $\alpha_H(a''|h_t) > 0$. 
	\end{enumerate}
\end{theorem}

Note that the attrition structure of signaling imposes strong restrictions on actions that can be played in equilibrium. Firstly, any separating action that perfectly identifies the low type must be myopically optimal for him, since the low type does not have any strategic incentives to play anything else. Secondly, if the low type mixes between pooling and separating, then he must be indifferent between the two: the gains from pooling (higher reputation) are exactly offset by the cost of taking suboptimal actions in current and/or future periods. 

It is worth emphasizing that the result holds under very minimal assumptions on payoffs and signals: the only requirements imposed on the model are that the sender's payoff is increasing in $p$ (which, in fact, is only required for the low type) and that the outcomes $x$ are not perfectly revealing. If the setting of interest fits this framework, then attrition is \emph{the only} informative equilibrium structure that can arise in this setting, unless one is willing to allow for NDOC-nonconformant beliefs off the equilibrium path. Under attrition, the high type is playing some pooling action, while the low type mixes between that and a separating action. 

One important case, which lies beyond the scope of our model, but is nonetheless worth mentioning, is that with a behaviorally committed type of the sender, and a strategic type, who prefers to mimic the committed type. For example, in the bargaining model of \citet{AG00}, a player may be either committed to rejecting all offers that give him anything less than the whole surplus, or fully strategic. Similarly, in the (static) cheap talk model of \citet{Chen11}, the sender may either be committed to truthful communication, or communicate strategically. Theorem \ref{thm:2types} applies to such problems (with the exception of the payoff equivalence part of statement 1), since its proof only relies on the incentives of the low type -- which in these settings is the strategic type. It follows that if the committed type is unable to verifiably demostrate his commitment, perfect separation is impossible in equilibrium, and all equilibria feature either attrition of the strategic type as in \citet{AG00}, or full pooling.

\section{Finite Types} \label{sec:finite}

We now move to exploring the setting with more than two but finitely many types. In this section we show that the insight of Theorem \ref{thm:2types} can be extended to this case, although allowing for many types does raise a number of additional issues and calls for extra assumptions.

\subsection{Single-Crossing} \label{sub:sc}

In order to secure the result in case of many types, we need to impose the following new assumption on payoffs:
\begin{trivlist}
	\item
	\begin{description}
		\item[(SC)] 
		For all optimal $\mathbf{b}$, any $\mathbf{a}',\mathbf{a}'' \in \cup_{h_t \in \mathcal{H}} \cup_{\theta \in S(h_t)} \arg \max_{\mathbf{a}} U (\mathbf{a}|h_t, \mathbf{b}, \theta)$, and all $h_t \in \mathcal{H}$, function $\mathcal{U}(\theta) \equiv U (\mathbf{a}'' | h_t, \mathbf{b}, \theta) - U (\mathbf{a}' | h_t, \mathbf{b}, \theta)$ either crosses zero at most once, or is identically zero.
	\end{description}
\end{trivlist}

This assumption belongs to a family of single-crossing conditions widely encountered in the literature on signaling, monotone comparative statics, and mechanism design.\footnote{See \cite{LM02} from a contract theory perspective (e.g., Ch. 2.2.3). Classic references on MCS,  in turn, include \citet{MS94} and \citet{Athey02}.} 
The purpose of our condition is standard: to ensure that the agent's preferred strategy is, in some sense, monotone w.r.t. his type. There are, however, some distinctive features that differentiate it slightly from other single-crossing conditions in the literature. 

Firstly, (SC) is a condition on the expectation of a discounted sum $\mathbb{E} \sum_t e^{-rt} u(a_t,b_t,\theta)$ rather than on the flow utility $u(a,b,\theta)$. While the latter would be more preferable, aggegating single-crossing is not a trivial problem. \citet{QS12} discuss this problem and offer possible solutions, but none of them apply to our setting.
Secondly, (SC) is more demanding than might appear initially. The dependence of $U(\mathbf{a} | h,\mathbf{b},\theta)$ on $\mathbf{a}$ realizes not only directly -- through the effect of agent's own action $a_t$ on his flow utility $u(a_t,b_t,\theta)$ -- but also via an indirect reputation channel. The receiver's response $b_t$ depends on agent's reputation $p_t$, which is, in general, affected by the agent's action choice $a_t$. Further, this reputation effect is persistent, with the choice of $a_t$ affecting not only the contemporaneous response $b_t$, but also the continuation payoff: for a given fixed path $\{a_s,x_s\}_{s>t}$, reputation $\{p(h_s)\}_{s>t}$ will be persistently shifted by $a_t$, meaning the receiver's responses $\{b_s\}_{s>t}$ are affected.

All of the above means that (SC) is a non-trivial condition and may be difficult to verify in some settings. If anything, verifying (SC) might as well be the main impediment to exploiting this paper's results in applied models. However, this task is far from impossible, with Section \ref{sec:app} demonstrating a number examples of applied models that can be easily verified to satisfy (SC).

\subsection{Attrition Structure of Equilibrium Signaling} \label{sub:attrn}

Theorem \ref{thm:main} that we gradually build up to is the analog of Theorem \ref{thm:2types} for the case when $|\Theta| > 2$, in the sense of characterizing the actions available in equilibrium at any history.
We begin, however, by stating a weaker result which, by looking at strategies rather than actions, provides a clearer characterization of the attrition structure of equilibrium signaling with $|\Theta| > 2$. Proposition \ref{prop:attrn} below establishes that as long as (SC) and other previously stated assumptions hold, strategies played in an arbitrary equilibrium of the game can be split into two classes. The first class consists of pooling strategies played by all types. While a nominal multiplicity of such strategies may arise, they must all be payoff-equivalent, so this class is, in a sense, degenerate. The second class is that of separating strategies employed by the lowest type -- these may vary in which pooling strategies they mimic and for how long. However, any separating strategy is only played by the lowest type.

To state this and other results we need to introduce some additional notation and definitions. Firstly, denote the two boundaries of the belief support at a given history $h_t$ as $\bar{S}(h_t) \equiv \max S(h_t)$ and $\underline{S}(h_t) \equiv \min S(h_t)$ respectively.
Furthermore, in a manner similar to type support $S$, given an equilibrium strategy profile $\{\alpha_\theta\}$, let us define action support as 
\begin{align*}
	A(h_t) \equiv\cup_{\theta \in S(h_t)} \left\{ a \in A \mid \alpha_\theta(a|h_t) > 0 \right\}.
\end{align*}
We say that given the receiver's strategy $\mathbf{b}$, a pure strategy $\mathbf{a}$ \emph{arrives at history $h_t$} -- and denote it as $\mathbf{a} \curlywedge h_t$ -- if $\mathbf{a}(h_{\tau}) = a_{\tau}(h_\tau)$ for all $h_{\tau}$ s.t. $h_t \succ h_\tau$.
Further, say that $\mathbf{a}$ is \emph{on path for $\theta$ at $h_t$} if $\mathbf{a} \curlywedge h_t$ and $\mathbf{a}$ is on path according to type $\theta$'s equilibrium strategy $\alpha$ starting from $h_t$: $\alpha_\theta(\mathbf{a}(h_t)|h_t) > 0$. Say that $\mathbf{a}$ is \emph{on path at $h_t$} if it is on path at $h_t$ for some $\theta \in S(h_t)$. 

We proceed by defining payoff equivalence of strategies in a straightforward manner.
\begin{definition}
	Fix an equilibrium $(\alpha,\mathbf{b},p)$ and history $h_t$. Any two pure strategies $\mathbf{a}',\mathbf{a}'' \curlywedge h_t$ are:
	\begin{itemize}
		\item \emph{payoff-distinct at $h_t$} if there exists $\theta \in S(h_t)$ such that $U(\mathbf{a}'|h_t,\mathbf{b},\theta) \neq U(\mathbf{a}''|h_t,\mathbf{b},\theta)$;
		\item \emph{payoff-equivalent at $h_t$} if they are not payoff-distinct at $h_t$.
	\end{itemize}
\end{definition}

The result can then be stated as follows.

\begin{proposition} \label{prop:attrn}
	Suppose the payoff function $u$ satisfies (MON) and $dt \to 0$.
	Fix an equilibrium $(\alpha,\mathbf{b},p)$ such that (NDOC-P) and (SC) hold. Fix some history $h_t \in \mathcal{H}$ and define $\underline{\theta} \equiv \underline{S} (h_t)$. Then for any pure strategy $\bar{ \mathbf{a}}'$ on path at $h_t$ for some type $\theta' \in S(h_t) \backslash \underline{ \theta}$, the following hold:
	\begin{enumerate}
		\item $\bar{ \mathbf{a}}'$ is optimal for all $\theta \in S(h_t)$ at $h_t$;
		\item any $\bar{ \mathbf{a}}''$ optimal for any $\theta'' \in S(h_t) \backslash \underline{ \theta}$ is payoff-equivalent at $h_t$ to $\bar{ \mathbf{a}}'$;
		\item there exists $\bar{ \mathbf{a}}'''$ that is payoff-equivalent at $h_t$ to $\bar{ \mathbf{a}}'$ and is on path for $\underline{ \theta}$ at $h_t$;
		\item any $\underline{\mathbf{a}}$ that is on path at $h_t$ and payoff-distinct at $h_t$ from $\bar{ \mathbf{a}}'$ is only on path for $\underline{ \theta}$.
	\end{enumerate}
\end{proposition}

To understand this proposition, it is illustrative to ignore payoff equivalence for a second and treat any pair of payoff-equivalent strategies as the same strategy. In this reading, the proposition implies that any pure strategy $\mathbf{a}$ on path for some type $\theta \in S(h_t)$ is on path for \emph{all} types $\theta \in S(h_t)$, including the currently-lowest type $\underline{ \theta}$. Therefore, no type of the agent can ever conclusively separate from $\underline{ \theta}$. At the same time, there may exist strategies that separate $\underline{ \theta}$ away from the remaining types. The weight that the receiver's belief assigns to $\underline{ \theta}$ may thus decrease over time along the pooling path of play, it may even converge to zero asymptotically as $t \to \infty$, but it may never become exactly zero.
However, the interpretation above is overly strong, since payoff-equivalent strategies do not need to coincide at all histories. In other words, it is a statement about \emph{strategies}, whereas we would like to have a result about \emph{actions}.

\subsection{From Strategies to Actions}

The question that remains unanswered by Proposition \ref{prop:attrn} is whether payoff-equivalence implies strategy-equivalence (i.e., that any two payoff-equivalent strategies must prescribe the same actions at either all, or at least some histories) or, if not, whether payoff-equivalent strategies at least produce equivalent belief paths for all types. Unfortunately, the answer to both of the above is negative: in general, not only may there be multiple payoff-equivalent strategies, but they may even induce different beliefs. This is demonstrated by the following example.

\begin{example}
	Suppose $\Theta = \{0,1,2\}$, types are ex ante equiprobable, $A = \mathbb{R}_+$, and outcomes are uninformative. 
	Suppose the sender's reduced-form utility function \eqref{eq:utilde} is given by $\tilde{u}(a,p,\theta) = \mathbb{E}_p(\theta)$ (so agent's actions are cheap talk; (SC) holds trivially in this scenario).
	Then the following strategies constitute an equilibrium together with respective beliefs: type $\theta=2$ plays strategy $\mathbf{a}'' = (a'',0,0,...)$, while types $\theta=1,3$ play $\mathbf{a}'=(a',0,0,...)$, where $a'\neq a''$ are arbitrary. 
	In this PBE some information about type is conveyed in period zero -- namely, type $\theta=2$ separates from $\theta=1,3$. However, all types of the sender are indifferent between the two strategies, hence information revealed by $a_0$ is not relevant to the sender's payoff -- although it may be relevant for the receiver. 
\end{example}

However, we are arguably more interested in \emph{payoff-relevant} signaling, which relies on the heterogeneity of the agent's preferences across types to convey information, as opposed to the agent's utmost indifference. Narrowing the focus to such payoff-relevant information revelation allows to carry the insight of Proposition \ref{prop:attrn} over from strategies to actions. We begin by stating the formal definitions of payoff-relevant and irrelevant signaling in our setting.

\begin{definition}
	Fix an equilibrium $(\alpha,\mathbf{b},p)$ and history $h_t \in \mathcal{H}$. 
	\begin{itemize}
		\item \emph{Payoff-relevant signaling} happens at $h_t$ if there exist $a',a'' \in A(h_t)$ and $\theta \in S(h_t)$ such that $V (a'|h_t, \mathbf{b}, \theta) \neq V (a''|h_t, \mathbf{b}, \theta)$.
		
		\item \emph{Payoff-irrelevant signaling} happens at $h_t$ if there exist $a',a'' \in A(h_t)$ such that $p(h_t,a',x) \neq p(h_t,a'',x)$ for some $x \in X$ but $V (a'|h_t, \mathbf{b}, \theta) = V (a''|h_t, \mathbf{b}, \theta)$ for all $\theta \in S(h_t)$.
	\end{itemize}
\end{definition}

In other words, payoff-relevant signaling implies that at a given history $h_t$ there are two distinct actions on path, $a'$ and $a''$, and there is some type of the agent for which the choice between these two actions has payoff consequences. Note that since both actions are on path, it cannot be the case that all types prefer one over another -- both $a'$ and $a''$ must be optimal for some types of the agent. Payoff-relevance of this action choice is then defined as some type $\theta \in S(h_t)$ having strict preference between the two.

We are now ready to state the theorem that characterizes payoff-relevant signaling in terms of actions, making the implications of Proposition \ref{prop:attrn} more explicit. The result below expands the message obtained in Theorem \ref{thm:2types} to the case of finitely many types, albeit at the cost of restricting model scope to payoff functions that satisfy (SC) and to the continuous-time limit of the model (as opposed to any sufficiently small $dt$).

\begin{theorem} \label{thm:main}
	Suppose the payoff function $u$ satisfies (MON) and $dt \to 0$.
	Fix an equilibrium $(\alpha,\mathbf{b},p)$ such that (NDOC-P) and (SC) hold. Fix some history $h_t \in \mathcal{H}$.
	If payoff-relevant signaling happens at $h_t$ then, defining $\underline{\theta} \equiv \underline{S} (h_t)$, the following hold:
	\begin{enumerate}
		\item any on-path action $a \in A(h_t)$ is on path for $\underline{ \theta}$ at $h_t$;
		\item $A(h_t) \cap A^*\left( h_t, \mathbf{b}, \underline{\theta} \right)$ is nonempty, and any $\underline{a}$ in the intersection is on path only for $\underline{ \theta}$ at $h_t$;
		\item any action $\bar{a} \in A(h_t) \backslash A^*\left( h_t, \mathbf{b}, \underline{\theta} \right)$ is optimal at $h_t$ for all $\theta \in S(h_t)$.
	\end{enumerate}
\end{theorem}

What the theorem says is that in any equilibrium with payoff-relevant signaling, there are effectively at most two types of \emph{actions} -- as opposed to \emph{strategies} in Proposition \ref{prop:attrn} -- on path at any history: pooling actions (typical element $\bar{a}$) and separating actions (typical element $\underline{a}$). The latter are only ever played by the currently-lowest type $\theta = \underline{S} (h_t)$ and separate him from the remaining types. As in Theorem \ref{thm:2types}, any separating action must be myopically optimal for the lowest type given that he is revealed. 

Pooling actions, on the other hand, are optimal for all types. Further, if no payoff-irrelevant signaling takes place, then any pooling action is, in fact, on path for all $\theta \in S(h_t)$ -- i.e., all types do actually pool on the pooling action(s). Notably, both payoff-relevant and payoff-irrelevant signaling may occur simultaneously at a given history. In that case there will be more than one pooling action, and while all of them are necessarily on path for $\underline{\theta}$, the higher types may vary in their action choices, despite all types being indifferent between all of these pooling actions. 

The corollary below relates to the situations when payoff-relevant signaling occurs at successive histories. It states that the pooling action in the earlier history must then be such that the low type is indifferent between separating and pooling -- meaning that \emph{flow payoffs} the low type gets from the separating and pooling actions must be the same. The low type must be indifferent between separating at $t$ and $t+dt$, so one period of pooling must be exactly as attractive as one period of being identified as $\underline{ \theta}$. In practice, this means that pooling action must be costlier for $\underline{ \theta}$ than the separating action, since the former yields higher reputation payoff.

\begin{corollary} \label{cor:repsig}
	Suppose the conditions in Theorem \ref{thm:main} hold. Suppose payoff-relevant signaling occurs also at $h_{t+dt} \equiv (h_t, (\bar{a},x,\mathbf{b}_t))$ for some $\bar{a}$ and all $x$ in the support. Then such $\bar{a}$ must satisfy $\mathbb{E}_x \left[ \tilde{u}(\bar{a},p(h_t, (\bar{a},x)),\underline{\theta}) | \underline{ \theta} \right] = \tilde{u}(\underline{a}, \delta_{\underline{\theta}},\underline{\theta})$, where $\tilde{u}(a,p,\theta)$ is defined by \eqref{eq:utilde}.
\end{corollary}

Finally, Theorem \ref{thm:main} applies to all histories, including those off the equilibrium path. Applying it inductively starting from the root history, we obtain Corollary \ref{cor:supp} below, which states that in the absence of payoff-irrelevant signaling, only the lowest type $L \equiv \min \Theta$ can ever separate from the rest, while the remaining ones can never separate from one another.

\begin{corollary} \label{cor:supp}
	Suppose the conditions in Theorem \ref{thm:main} hold. 
	In any equilibrium in which no payoff-irrelevant signaling happens, for any on-path history $h_t$, one of the following must hold: 
	\begin{enumerate}
		\item $S(h_t) = \Theta$;
		\item $S(h_t) = \{\min \Theta\}$.
	\end{enumerate}
\end{corollary}

It is worth noting that there may be histories $h_t$ at which $p(h_t)$ assigns arbitrarily small weight to the lowest type. So while this type can never be ruled out completely along the pooling path, asymptotically the receiver's belief may assign arbitrarily low weight to it. Notably, this implies that the mechanism of attrition of the lowest type can yield full separation asymptotically if there are only two types of the sender but not if there are more (but finitely many), which is an important takeaway, since many applied papers treat two-type models as proxies for more general settings. 

At the same itme, the assumption that is crucial to our analysis is that the lowest type is separated away from all other types. Otherwise -- e.g., with an interval type space -- full asymptotic revelation is possible again. An example of such outcome in the context of bargaining is presented by \citet{FS10}. Their equilibrium resembles the attrition equilibria of this paper, except in their equilibrium, an interval of lowest types separates away in every period instead of the single lowest type mixing between that and pooling. This observation serves to illustrate the nontrivial implications of model discretization: if one attempts to compute equilibria of a continuous-time signaling model with an interval of types by approximating it with a disrete-time finite-type model, different approximations may yield qualitatively different results. In particular, if an interval type space is approximated by a grid that is too coarse relative to time discretization, the researcher would conclude that no asymptotic learning takes place in the discretized model, whereas it could take place in the continuous model.

Theorem \ref{thm:main} and its corollaries effectively provide a cookbook on how to construct an equilibrium with payoff-relevant signaling only. Suppose we want signaling to occur during the time interval $[0,T]$. Then in every period, along the pooling path we shall have two actions available to the sender: a separating action $\underline{a} \in A^*\left( h_t, \mathbf{b}, L \right)$ only taken by the lowest type $L \equiv \min \Theta$ and a pooling action $\bar{a}$ that satisfies the condition in Corollary \ref{cor:repsig} -- the latter action will be played by $L$ with some probability and by all other types for sure. Note that we have a degree of freedom in this construction: reputation from taking a pooling action depends on the probability with which type $L$ separates in that given period. Hence by changing these probabilities we will be able to sustain different pooling actions $\bar{a}$ in equilibrium. To complete the construction, we need to verify that from time $T$ onwards, the pooling strategy is such that $L$ is exactly indifferent at $T$ (or the last period before $T$) between separating and following this pooling path, and to verify that all other types always weakly prefer the pooling action to the optimal deviation. Appendix B provides an example of an equilibrium constructed using this cookbook, in the context of the price signaling model developed in the following section.

\section{Applications} \label{sec:app}

This section presents examples of applied models in different settings that fit our framework. It is meant to demonstrate some instances of models yielding additively and/or multiplicatively separable payoff functions that allow (SC) to be verified with little effort. 
We begin by presenting in Section \ref{sub:suff} a specific separable framework and show that there exist simple sufficient conditions for (MON) and (SC) within this framework.
We then proceed to applying this framework, starting with a model for the example from the Introduction, in which a firm can use its historic prices to convince the antitrust authority that the firm has no monopoly power. Section \ref{sub:antitrust} sets up the model and verifies (both directly and using results from Section \ref{sub:suff}) that our results apply to it; Appendix B constructs an example of an informative equilibrium to verify that they exist and to show more concretely how they can look. Other applications, to labor market signaling and bargaining, are considered in Sections \ref{sub:jmsig} and \ref{sub:barg} respectively.

\subsection{Separable Settings} \label{sub:suff}

This section shows that if the agent's flow utility function $\tilde{u}(a,p,\theta)$ is separable in a specific way, then (MON) and (SC) can be easily verified. While this form of the utility function may appear restrictive, the remainder of Section \ref{sec:app} shows that it captures a wide range of settings, including classic signaling and bargaining models.

In particular, suppose the agent's flow utility function $\tilde{u}$ defined in \eqref{eq:utilde} can be represented as 
\begin{equation}
	\tilde{u}(a,p,\theta) = \phi_0(a,p) + \phi_1(a,p) \psi(\theta)
	\label{eq:u_repr}
\end{equation}
for some collection of functions $\phi_0, \phi_1, \psi$.
Then we can derive simple conditions on these three functions that are sufficient for (MON) and (SC) to hold. These conditions are given by the two respective propositions below.

\begin{proposition}
	\label{prop:suff_mon}
	If representation \eqref{eq:u_repr} applies, with $\psi(\theta) \geq 0$ and $\phi_0(a,p),\phi_1(a,p)$ weakly increasing in $p$, then $\tilde{u}$ satisfies (MON).
\end{proposition}

As in the rest of the paper, monotonicity in $p$ is understood with respect to $\geq_{FOSD}$ order on $p$ (see footnote \ref{foot:fosd}).
Note that $\phi_1(a,p) \psi(\theta)$ can be negative for some or all $a,p,\theta$; without loss we let $\phi_1$ absorb the negative sign in this case.

\begin{proposition}
	\label{prop:suff_sc}
	If representation \eqref{eq:u_repr} applies and outcomes $x_t$ are uninformative at all $h_t$, then (SC) holds if $\psi(\theta)$ is strictly monotone in $\theta$.
\end{proposition}

We now continue to the more specific applications that use these results.

\subsection{Price Signaling} \label{sub:antitrust}

In this section we revisit the example from the Introduction, in which Apple is using its past pricing decisions to argue that it faces competitive pressure with regards to its App Store. This section constructs a simple model for this story that fits the framework of Section \ref{sec:model}, thereby demonstrating that past prices can not serve as conclusive proof of a lack of monopoly power. We then construct an informative equilibrium in Appendix B for the special case of this model, showing that prices \emph{can} serve as \emph{suggestive} evidence.

To construct the simplest model possible, let us adopt a framework in the spirit of monopolistic competition. Consider a firm (Apple) that serves app developers, and for simplicity ignore the downstream market, in which developers interact with app users. The firm faces residual demand curve $q_t=(1-\theta a_t) dt$ in every period $t \in \mathcal{T} \equiv \{0, dt, 2dt, ... \}$, where $a_t \in \mathbb{R}_{+}$ is the price (App Store commission) the firm sets in that period, and $\theta \in \Theta \subseteq \mathbb{R}_{++}$ is the degree of the competitive pressure that the firm faces, hereinafter referred to as the state. In this specification, if Apple App Store is, in fact, competing with Google Play Store and other app and game stores on other devices, this is reflected by higher $\theta$ and lower residual demand for Apple's services.

In every period, after the firm sets price $a_t$, the developers may file a complaint to the antitrust authority (regulator). This happens with probability $\lambda(a_t) dt$, which is weakly increasing in $a_t$. If a complaint is filed, the regulator opens an investigation, which results in a fine of size $F$ if evidence indicates that the competitive pressure $\theta$ is low. In particular, assume that if an investigation is launched, a fine is imposed with probability $\gamma(\mathbb{E}[\theta|h_t])$ that is strictly decreasing in the expectation $\mathbb{E}[\theta|h_t]$ of the state inferred from some prior belief $p_0 \in \varDelta(\Theta)$ and the history of the firm's past pricing choices $h_t = (a_0, ..., a_{t-dt})$.

To verify that Theorems \ref{thm:2types} and \ref{thm:main} apply to this model, we need to verify that assumptions (MON) and, in case of Theorem \ref{thm:main}, (SC) hold. 
The flow payoff function can be written down as
\begin{align} \label{eq:antitrust_flow}
	\tilde{u}(a,p,\theta) = a (1 - \theta a) - \lambda(a) \gamma(p) F,
\end{align}
where $\gamma(p_t) = \gamma(p(h_t)) \equiv \gamma(\mathbb{E}[\theta|h_t])$. 
This utility function fits representation \eqref{eq:u_repr} with $\phi_0(a,p) = a - \lambda(a) \gamma(p) F$, $\phi_1(a,p) = -a^2$, and $\psi(\theta) = \theta$. One can see that $\psi(\theta) > 0$ is strictly increasing, and $\phi_0(a,p), \phi_1(a,p)$ are weakly increasing in $p$, since $\phi_0(a,p)$ is strictly decreasing in $\gamma$, which is strictly decreasing in $\mathbb{E}[\theta|p]$, which is strictly increasing in $p$ w.r.t. FOSD shifts. Hence Propositions \ref{prop:suff_mon} and \ref{prop:suff_sc} apply, and (MON) and (SC) hold in this model.

Both assumptions can be verified directly. To verify (MON), note that \eqref{eq:antitrust_flow} only depends on $p$ through $\gamma(p)$, and as argued above, $\gamma$ is strictly increasing in $p$ w.r.t. FOSD shifts, hence (MON) holds.
To verify (SC), fix some $\mathbf{a}', \mathbf{a}''$; then the function $\mathcal{U}(\theta) \equiv U (\mathbf{a}'' | h_t, \theta) - U (\mathbf{a}' | h_t, \theta)$ for some fixed $h_t$ is given by:
\begin{align*}
	\mathcal{U}(\theta) &= \sum_{s \in \mathcal{T}, s \geq t} e^{-r(s-t)} \bigg( \mathbb{E} \left[ \mathbf{a}''_s \left(1- \theta \mathbf{a}''_s \right) - \lambda(\mathbf{a}''_s) \gamma(p_s) F \mid h_t, \mathbf{a}'' \right] -
	\\
	& \phantom{= \sum_{s \in \mathcal{T}, s \geq t} e^{-r(s-t)} \bigg(}
	- \mathbb{E} \left[ \mathbf{a}'_s \left(1- \theta \mathbf{a}'_s \right) - \lambda(\mathbf{a}'_s) \gamma(p_s) F \mid h_t, \mathbf{a}' \right] \bigg)
	\\
	&= C_1 + C_2 \theta
\end{align*}
for some constants $C_1, C_2$ that depend on $\mathbf{a}',\mathbf{a}'',h_t$. We see that $\mathcal{U}(\theta)$ is a linear function of $\theta$, hence (SC) holds.

The results of Theorems \ref{thm:2types} and \ref{thm:main} therefore apply: the firm can \emph{not} prove conclusively, by referring to prices alone, that its market power is sufficiently low (i.e., that $\theta$ is below some threshold $\bar{\theta}$ set by the regulator). However, prices \emph{can} serve as suggestive evidence. To demonstrate this, an informative equilibrium is constructed in Appendix B given specific functional forms for $\lambda$ and $\gamma$.

It should be self-evident that the goal of this section is to produce the simplest model for the setting. As a result, the model reduces the effects of competition to a residual demand curve, reduces consumers to non-strategic complainers, assumes the antitrust action is just a fixed fine, that investigations do not condition on past investigations, etc. A paper aiming to explore this particular phenomenon could set up a more convincing model that avoids the aforementioned simplifications.

\subsection{Labor Market Signaling} \label{sub:jmsig}

In this section we revisit the classic labor market signaling model (\citet{Spe73}), which sparked the original discussion around dynamic signaling (\citet{NVD90}, \citet{Swi99}). In the dynamic version of this model, a long-lived candidate of privately known ability $\theta \in \Theta \subseteq \mathbb{R}_+$ acquires costly and, w.l.o.g., unproductive education in an attempt to signal her ability to potential employers. A high-ability worker is more productive on the job and can thus bargain for a higher wage, while also having lower cost of education than a low-ability worker.
In every period $t \in \mathcal{T} \equiv \{0, dt, 2dt, ... \}$ she chooses education intensity $e \in E \subset \mathbb{R}$. The flow cost of education is given by $c(e|\theta) \equiv l(e) \cdot m(\theta)$, where $l(e)$ is increasing in $e$ with $l(0)=0$, and $m(\theta)$ is strictly decreasing in $\theta$. 

There is a population of homogeneous competitive employers, who observe the full history of the candidate's education choices and grades. In every period they simultaneously offer employment contracts to the candidate.\footnote{Suppose that the offers are made privately and so are not observed by other firms. \citet{NVD90} and \citet{Swi99} show that otherwise -- if the offers are public, -- a tacit collusion equilibrium with perfect separation can be sustained, see footnote \ref{foot:NVD_Swi} for details.} 
After observing all contracts, the candidate may accept at most one of them. If a contract is accepted, in every future period the candidate receives wage $w\cdot dt$, where $w$ is as specified in the contract. Let $d \in \{0,1\}$ denote the worker's acceptance decision -- whether she chooses to accept an offer in a given period or not. If the candidate chooses to accept, she would trivially find it optimal to choose the highest-wage contract.

W.l.o.g., let $\theta$ be equal to the candidate's on-the-job productivity (so her output is $\theta \cdot dt$ per period). This means that at any history $h_t$, all competitive firms will offer the same wage $w(h_t) = \mathbb{E} [\theta | h_t, d(h_t)=1]$. A history here consists of the candidate's past actions: $h_t = \{d_s,e_s\}_{s\in \mathcal{T}, s<t}$. The implied timing in the stage game at any history $h_t$ is:
\begin{enumerate}
	\item if the candidate has accepted an offer with wage $w$ in the past, she receives the contracted wage $w\cdot dt$, and the game proceeds to the next period. Otherwise,
	\item firms make wage offers $w(h_t)$ to the candidate;
	\item the candidate decides $d(h_t)$ whether to accept the highest-wage contract. If $d(h_t)=1$ then the game continues to the next period. Otherwise,
	\item the candidate chooses $e(h_t)$, her education effort in the current period, and the game continues to the next period.
\end{enumerate}
Once a candidate has accepted an offer, at all future histories set $d(h_t)=e(h_t)=0$.

In such a game, the candidate's payoff from following some given strategy $\mathbf{a} = \{d(h),e(h)\}_{h\in \mathcal{H}}$ conditional on some history $h_t \in \mathcal{H}$ at which she has not yet accepted an offer is given by
\begin{align*}
	U (\mathbf{a} | h_t, \theta) \equiv& \sum_{s \in \mathcal{T}, s \geq t} e^{-r(s-t)} \left[ d(h_s) \frac{w(h_s)}{r} - \left( 1 - d(h_s) \right) \cdot c\left(e(h_s)|\theta \right) dt \right] dt,
\end{align*}
since accepting an offer at $h_t$ is equivalent to receiving a lumpsum payoff of $w(h_t)/r$.
The flow utility function can thus be framed as \eqref{eq:u_repr} using functions $\phi_0((d,e),p) = d \frac{w(p,d)}{r}$, $\phi_1((d,e),p) = -(1-d)l(e)$, and $\psi(\theta) = m(\theta)$. Here $m(\theta)$ is strictly monotone, $\phi_1((d,e),p)$ is independent of $p$, and $\phi_0((d,e),p)$ is weakly increasing in $p$, since $w(p,d) = \mathbb{E}[\theta|p,d=1] = \mathbb{E} \big[ \mathbb{E}[\theta | p] \mid d=1 \big]$ by the law of iterated expectations, and $\mathbb{E}[\theta | p]$ is strictly increasing in $p$. Therefore, by Propositions \ref{prop:suff_mon} and \ref{prop:suff_sc}, (MON) and (SC) hold.

Verifying (MON) directly is not much different from invoking Proposition \ref{prop:suff_mon} above, same as in Section \ref{sub:antitrust}. It is easy to verify (SC) directly as well: for any pair of strategies $\mathbf{a}',\mathbf{a}''$, the function $\mathcal{U}(\theta) \equiv U (\mathbf{a}'' | h_t, \theta) - U (\mathbf{a}' | h_t, \theta)$ can be written as 
\begin{align*}
	\mathcal{U}(\theta) =& \sum_{s \in \mathcal{T}, s \geq t} e^{-r(s-t)} \bigg[ \left(d'(h_s)-d''(h_s)\right) \frac{w(h_s)}{r} -
	\\
	& - \big[ \left( 1 - d'(h_s) \right) \cdot l\left(e'(h_s)\right) - \left( 1 - d''(h_s) \right) \cdot l\left(e''(h_s)\right) \big] m(\theta) dt \bigg],
	\\
	=& C_1 + C_2 m(\theta)
\end{align*}
for some $C_1,C_2$ that depend on $\mathbf{a}', \mathbf{a}'', h_t$.
Since $m(\theta)$ is strictly decreasing in $\theta$, $\mathcal{U}(\theta)$ is either strictly monotone, or constant, depending on whether the coefficient at $m(\theta)$ is positive, negative, or zero. Since strategies $\mathbf{a}',\mathbf{a}''$ and history $h_t$ were arbitrary, this means that payoff function $U(\mathbf{a}|h, \theta)$ satisfies (SC) and our results apply. The candidate will only be able to signal her ability via attrition: in an informative equilibrium -- if it exists -- the lowest type would drop out of education and take up a job at a random date, while all other types go through the whole education process prescribed by the equilibrium and only accept the job afterwards. The latter path would also be taken by some low-ability candidates. 
\citet{Swi99}, however, shows that pooling is the only equilibrium in this model, i.e., even the attrition-form informative equilibria do not exist.

\subsection{Bargaining} \label{sub:barg}

In this section we consider a simple model of bilateral bargaining, in which one party's valuation is commonly known, while another party's valuation is their private information. 
Consider two players, a buyer $B$ and a seller $S$, who interact repeatedly in every period $t \in \mathcal{T} \equiv \{0, dt, 2dt, ... \}$. There is a unit of indivisible good of some quality $\theta \in \Theta \subseteq \mathbb{R}_+$ privately known by the seller. The buyer has some commonly known prior belief $p_0 \in \varDelta(\Theta)$ about quality $\theta$, but does not observe the realization of $\theta$. The seller is initially in possession of the item and values it at $c(\theta)$, assumed to be either constant, or strictly increasing. The buyer's valuation is $v(\theta)$. In every period one of the players $x_t \in \{B,S\}$ is chosen as the proposer (the choice rule can be random or deterministic, and/or history-dependent) and can offer a price $y_t \in \mathbb{R}_+$. The other player then decides $z_t \in \{0,1\}$ whether to accept the offer. If the offer is accepted, the item is traded at that price and the game ends (which can be emulated by setting $y=z=0$ for both players at all subsequent histories). Otherwise the game continues to the next period. Both players discount the future at rate $r$.

History $h_t  = \{x_s,y_s,z_s\}_{s\in \mathcal{T}, s<t}$ in this game is given by the players' past offers and responses, as well as identities of the proposing player. The players' pure strategies are given by $\mathbf{a} = \{y_S(h),z_S(y,h)\}_{h\in \mathcal{H}, y\in \mathbb{R}_+}$ for the seller (conditional on type $\theta$) and $\mathbf{b} = \{y_B(h),z_B(y,h)\}_{h\in \mathcal{H}, y\in \mathbb{R}_+}$ for the buyer respectively, where $y_i(h)$ is the player's proposal if they are selected, and $z_i(y,h)$ is their response to the opponent's proposal $y$.\footnote{Note that regardless of $x_t$, the seller in this example always acts as ``the sender'' and the buyer as ``the receiver''. In particular, the setup of Section \ref{sec:model} implies that in this setting, the seller chooses the time-$t$ action $(y_S(h_t),z_S(y,h_t))$ before knowing time-$t$ role allocation $x_t$, and the buyer chooses $(y_B(h_t),z_B(y_S(h_t),h_t))$ after learning both $x_t$ and the seller's action. However, this deviation from the standard offer-response procedure is not meaningfully impactful when the buyer is restricted to Markov strategies.} 
Since the buyer is neither short-lived, nor myopic, this setting does not fall within the scope of the model defined in Section \ref{sec:model}. However, as mentioned in footnote \ref{foot:myopy}, those assumptions are only needed to shut down the folk theorem effects, so that (MON) is a non-vacuous assumption. In the context of this bargaining model, this can be achieved by assuming instead that the buyer follows a \emph{Markov} strategy $\{y_B(p), z_B(y_S,p)\}$, which treats belief $p_t$ as a sufficient statistic of the whole history $h_t$, and that both $y_B(p)$ and $z_B(y_S,p)$ are increasing in $p$ (where monotonicity is interpreted in the same semi-strong sense w.r.t. the F.O.S.D. order on $p$ as in the (MON) assumption).

Then denoting the probability with which the buyer is selected to propose at $h_t$ as $\chi(h_t) \equiv \mathbb{P} \left(x(h_t)=B \mid h_t \right)$, the seller's expected payoff from following some strategy $\mathbf{a}$ conditional on the buyer's strategy $\mathbf{b}$ and some history $h_t \in \mathcal{H}$, by which no offer had been accepted, is given by
\begin{equation}
	\begin{aligned} 
	U(\mathbf{a}|h_t,\theta,\mathbf{b}) \equiv \mathbb{E} \Bigg[ \sum_{s \in \mathcal{T}, s \geq t} e^{-r(s-t)} &\big[ \chi(h_s) z_S(y_B(p_s),h_s) \left( y_B(p_s)- c(\theta) \right) + 
	\\
	& + \left(1-\chi(h_s)\right) z_B(y_S(h_s),p_s) \left( y_S(h_s) - c(\theta) \right) \big] \Bigg].
	\end{aligned}
	\label{eq:U_barg}
\end{equation}
It is easy to see that $U(\mathbf{a}|h_t, \theta, \mathbf{b})$ is linear in $c(\theta)$, hence so is $\mathcal{U}(\theta) \equiv U(\mathbf{a}''|h_t,\theta,\mathbf{b}) - U(\mathbf{a}'|h_t, \theta, \mathbf{b})$ for any pair of seller's strategies $\mathbf{a}',\mathbf{a}''$ and any given buyer's strategy $\mathbf{b}$. Therefore, (SC) holds in this model as long as $c(\theta)$ is constant or strictly increasing, regardless of $v(\theta)$. 
Since both $y_B(p)$ and $z_B(y_S,p)$ are assumed to be increasing in $p$, (MON) holds as well.\footnote{Representation \eqref{eq:utilde} assumes flow utility only depends on $x_t$ via $p(h_t,a_t,x_t)$, whereas the flow utility in \eqref{eq:U_barg} explicitly depends on $x_t$ -- or, after taking expectations, on its distribution $\chi(h_t)$. However, in common settings of buyer-proposing ($\chi(h_t) \equiv 1$), seller-proposing ($\chi(h_t) \equiv 0$), and even alternating-offer ($\chi_0 \in \{0,1\}$, $\chi(h_{t+dt}) = 1 - \chi(h_{t})$) bargaining, we can treat $\chi(h_t)$ as an exogenous parameter.}

Alternatively, we could once again invoke Propositions \ref{prop:suff_mon} and \ref{prop:suff_sc}, since \eqref{eq:U_barg} can, given the buyer's equilibrium strategy $\{y_B(p), z_B(y_S,p)\}$, be represented in terms of \eqref{eq:u_repr} with $\phi_0((y_S,z_S),p | \chi) = \chi z_S y_B(p) + (1-\chi) z_B(y_B,p) y_S$ and $\phi_1((y_S,z_S),p | \chi) = -\chi z_S - (1-\chi) z_B(y_B,p)$ being monotone in $p$ and $\psi(\theta) = c(\theta)$ being strictly monotone.\footnote{If $c(\theta)$ is constant, Proposition \ref{prop:suff_sc} can be applied by setting $\tilde{u}(a,p,\theta) = \phi_0(a,p)$.}

The main implication of Theorems \ref{thm:2types} and \ref{thm:main} for this bargaining model is that the perfectly efficient allocation is unattainable for any finite delay. However, our results make no statements on whether the Coase conjecture is realized and all seller types trade at $t=0$ at the lowest acceptable price, or delay can be used as an (imperfect but effective) screening device.\footnote{There is a subset of literature exploring possible reasons for delay in bargaining (e.g., \citet{AG00} and \citet{FS05}). Our focus is different: instead of demonstrating sufficient conditions within the bargaining models under which \emph{delay is the only equilibrium}, this paper considers a much more general class of models and aims to provide weaker sufficient conditions under which \emph{attrition is the only informative equilibrium}. In particular, when applied to bargaining, our model does not rule out the possibility of instant agreement, in which all types of the sender pool on the same action.}
The exact shape that the equilibria can take depends on $c(\theta)$, $v(\theta)$, and $x(h_t)$. For example, if $c(\theta)=c$ for some constant $c$ and $v(\theta) \geq c + g$ for some ``gap'' value $g$ and the uninformed buyer makes all the offers -- then the Coase conjecture realizes, and in the unique equilibrium all types of the seller pool on accepting the lowest price (\citet*{GSW86}). In the alternating-offer scenario with the same assumptions, on the other hand, a delay equilibrium is possible, where attrition occurs for some finite time. During that time, some of the lowest-type sellers sell, and after that all the remaining sellers pool on the same price (\citet{AD98}, as presented in \citet*{ACD02}). If $v(\theta)=v \geq c(\theta)$ for some constant $v$ and all $\theta$, then in the seller-proposing game pooling is, again, the only equilibrium, but not in the Coase conjecture sense. In that equilibrium all types of the seller propose $v$ in every period and obtain surplus $v-c(\theta)$, with their offer being accepted straight away (\citet{AD89b}). For a review of these and other results on bargaining under incomplete information, see \citet{ACD02}.
As applied to more recent literature, our framework also subsumes the model of \citet{DG20}, who explore a setting in which public news arrive during the bargaining process.

\section{Conclusion} \label{sec:concl}

This paper explores a model of dynamic signaling with observable actions. In this model a single privately-informed agent takes an action every period, but cannot commit to future actions. The receiver tries to infer the agent's information from his actions, and the receiver's opinion is relevant to the agent's payoff. 
The existing literature has implied that signaling is impossible in such setting, unless strong assumptions about off-equilibrium-path beliefs are adopted. This paper confirms the negative result that \emph{perfect separation} is impossible in such a setting. However, it provides a novel positive result, showing that \emph{imperfect signaling} is possible under reasonable off-path beliefs. Further, we show that such signaling must necessarily happen through attrition of the lowest type of the agent. In this attrition scenario, all types pool on the same action (or split across a number of different yet payoff-equivalent actions), while the lowest type also plays some separating action with positive intensity. 

The paper identifies sufficient conditions, under which the results hold. These include a restriction on monotonicity of the agent's preferences w.r.t. his reputation, and a restriction on the off-path beliefs to be reasonable. In the case of many types single-crossing of agent's preferences must also hold. The latter is, arguably, the strongest of the three assumptions and the one that would be most difficult to verify in the applied work. However, the paper presents a number of applied signaling models to demonstrate that this notion of single-crossing can be used in applied work. Future work could involve the exploration of simpler notions of single-crossing that would work for dynamic signaling games.

Importantly, the paper assumes the receiver to be myopic. The main issue that arises when both the agent and the receiver are strategic is the folk theorem, which says that any individually rational payoff for either player can be sustained in equilibrium for $dt$ small enough. The consequence of this equilibrium multiplicity is that (MON) and (SC) almost never hold across the whole spectrum of equilibria in a given setting. The solution, if one wishes to explore settings with a strategic long-lived receiver, is to focus on some selected equilibria -- i.e., to restrict attention to some fixed strategy (or a class of strategies) $\mathbf{b}$ of the receiver and to test (MON) and, if necessary, (SC) against those strategies. This approach has been demonstrated in this paper in the bargaining application. Exploration of the specific equilibrium selection criteria that could yield favorable results lies beyond the scope of this paper, but could be another prospective direction for future research.

\bibliographystyle{abbrvnat}
\bibliography{srds_lit}

\renewcommand{\baselinestretch}{1.33}\small
\renewcommand{\thesection}{Appendix \Alph{section}.}
\renewcommand{\thesubsection}{\Alph{section}.\the\value{subsection}}
\setcounter{section}{0}
\setcounter{equation}{0}
\setcounter{figure}{0}
\renewcommand*{\theHsection}{app.\the\value{section}}
\renewcommand{\theequation}{\Alph{section}.\arabic{equation}}
\renewcommand{\thefigure}{\Alph{section}.\arabic{figure}}

\section{Proofs and Supplementary Results} \label{sec:proofs}

\subsection{Proofs: Preliminaries}

The first observation states that once there is no need for signaling any more -- i.e., when the receiver's belief assigns probability $1$ to some type of the agent -- there are no reasons for the agent to steer away from the myopically optimal action.

\begin{lemma} \label{lem:obs1}
	For any $dt$, in any equilibrium that satisfies (NDOC),
	at any $h_t \in \mathcal{H}$,
	if $|S(h_t)| = 1$ then for all $\theta$ and all $h_s \succeq h_t$: $\alpha_\theta\left( A^*(h_s, \mathbf{b}, \theta ) \mid h_s \right) = 1.$
\end{lemma}

\begin{proof}
	By (NDOC), for all $h_s \succeq h_t$: $p(h_s) = \delta_{S (h_t)}$. In particular, $p(h_s)$ is independent of all actions and outcomes during $[t,s)$, hence the receiver's best response $\mathbf{b} (h_s,a_s,x_s)$ is the same at all such $h_s$ (given $a_s$ and $x_s$). Therefore, the solution to \eqref{eq:util} is given by pointwise maximization of the flow utility.
\end{proof}

Lemma \ref{lem:obs1} above is the direct consequence of (NDOC): actions cannot change a degenerate belief under this assumption, hence the myopic optimum is chosen. This captures the main tension between signaling and sequential rationality: signaling requires sticking to the costly action over an extended period of time, while sequential rationality as captured by Lemma \ref{lem:obs1} pushes against that when no further signaling concerns are present.
The remaining statements formalize this intuition. However, before proceeding any further, we use Lemma \ref{lem:obs1} to prove Lemma \ref{lem:cond} from the text.

\begin{proof}[Proof of Lemma \ref{lem:cond}.]
	Denote the original equilibrium as $(\alpha_1,\mathbf{b}_1,p_1)$. Construct the new equilibrium $(\alpha_2,\mathbf{b}_2,p_2)$ by copying the strategies and beliefs the original equilibrium prescribes for all on-path histories $h_t$. For all off-path histories $h_t$, set $p_2(h_t) = \delta_{\underline{S}(h)}$, where $h$ is the last on-path history preceding $h_t$. For the receiver's strategy $\mathbf{b}_2(h_t,\alpha_2(h_t),x_t)$ at off-path $h_t$, take any optimal strategy that is consistent with on-path play w.r.t. the ``no sunspots'' assumption. Finally, the agent's strategy profile $\alpha_2(h_t)$ for off-path histories $h_t$ is set in conformance with Lemma \ref{lem:obs1}.
	
	Belief profile $p_2$ will then satisfy (NDOC-P) and be consistent with the strategy profile $\alpha_2$. The receiver's strategy $\mathbf{b}_2$ is, by construction, optimal at all histories. The agent's strategies $\alpha_{2,\theta}$ will be optimal at off-path histories by Lemma \ref{lem:obs1}. Optimality of $\alpha_{2,\theta}$ for type $\theta$ at any on-path history $h_t$ can be verified by observing that value $V(a | h_t, \mathbf{b}_2, \theta)$ is the same as in the original equilibrium for all on-path actions $a \in A(h_t)$ and weakly smaller for off-path actions $a \in A \backslash A(h_t)$ (since any such action generates a pointwise lower path of future reputation and (MON) holds). I.e., the choice between any pair of on-path actions is unaffected by the off-path modifications, while deviations to off-path actions are less appealing in the new equilibrium. We conclude that $(\alpha_2,\mathbf{b}_2,p_2)$ is an equilibrium.
\end{proof}

\subsection{Proofs: Two Types}

\begin{proof}[Proof of Theorem \ref{thm:2types}.]
	\textbf{Statement 1.} Suppose first, by way of contradiction, that there exist $h_t \in \mathcal{H}$ and $a' \in A$ such that $\alpha_H(a'|h_t) > 0$ but $\alpha_L(a'|h_t) = 0$. Then $p(h_t, a',x_t) = \delta_H$ for any $x_t \in X$ by (NDOC-P). By playing $a'$ at $h_t$ the low type receives the highest possible continuation utility after $t$ (since by Lemma \ref{lem:obs1} he can play the myopically optimal action thereafter), while by following the equilibrium path he receives strictly less due to (MON). The utility is bounded, hence for $dt$ small enough deviating to $a'$ at $h_t$ is optimal for $L$ -- a contradiction.\footnote{Note that by belief consistency and rationality there must exist $a'' \in A$ s.t. $\alpha_L(a''|h_t) > 0$ and $p(h_t,a'',x_t)(H) < p(h_t,a',x_t)(H)$ for all $x_t$, hence $V((h_t,a',x_t),\mathbf{b},L)$ is bounded away from $V((h_t,a'',x_t),\mathbf{b},L)$ for all $h_t, x_t, \mathbf{b}$. Therefore, there exists $\bar{dt}(h_t, \mathbf{b})$ s.t. if $dt < \bar{dt}(h_t, \mathbf{b})$ then period-$t$ gains for $L$ from playing $a''$ compared to $a'$ cannot outweigh the losses from $t+dt$ onwards.}
	
	Payoff-equivalence is shown as follows: for any two $a',a'' \in A$ such that $\alpha_H(a'|h_t) > 0$ and $\alpha_H(a''|h_t) > 0$ it must be that $V (a' | h_t, \mathbf{b}, H) = V (a'' | h_t, \mathbf{b}, H)$, otherwise the high type would only play one of the actions and not the other. The first part of the argument showed that $\alpha_L(a'|h_t) > 0$ and $\alpha_L(a''|h_t) > 0$, hence $V (a' | h_t, \mathbf{b}, L) = V (a'' | h_t, \mathbf{b}, L)$ by the same logic.
	
	\textbf{Statement 2.} Begin with the first part (that $a' \in A^*\left( h_t, \mathbf{b}, L \right)$). For any such $a'$ that $\alpha_H(a'|h_t) = 0$ and $\alpha_L(a'|h_t) > 0$ and any outcome $x_t$, we have $p(h_{t+dt}) = \delta_L$, where $h_{t+dt} = (h_t, a',x_t)$. By Lemma \ref{lem:obs1}, at all $h_s \succ h_t$, only bliss actions are played: $a_s \in A^*( h_t, \mathbf{b}, L)$. If $a' \notin A^*( h_t, \mathbf{b}, L)$ then playing a bliss action $a'' \in A^*( h_t, \mathbf{b}, L)$ at $h_t$ instead -- and continuing with $a_s$ at all subsequent histories -- yields a strictly higher flow payoff at $h_t$ and the same continuation payoff. Hence playing $a'$ at $h_t$ was not optimal.
	
	The second part of the second statement follows from the same argument as did payoff equivalence for $L$ in the first statement.
\end{proof}

\subsection{Proofs: Finite Types}

Before proceeding to the proof of Theorem \ref{thm:main}, it is convenient to split parts of it off into supplementary lemmas.
We begin by arguing in Lemma \ref{lem:supp1} that at no history can actions lead to separation of types into disjoint sets that can be compared by a strong set order -- unless one of these sets is a singleton coinciding with the lower bound of the other set. In particular, we show that sets of types in the support of two different actions have to necessarily overlap (not in the sense of having common elements, but in the sense of upper and lower bounds).

\begin{lemma} \label{lem:supp1}
	Suppose (MON) holds and $dt \to 0$. Fix any equilibrium and any history $h_t \in \mathcal{H}$.
	Then for any $a',a'' \in A(h_t)$ we have $\bar{S}(h_t, a') \geq \underline{S}(h_t, a'')$, with equality only if $S(h_t, a')$ is a singleton.\footnote{This Lemma and the remainder of the Appendix uses $S(h_t, a)$ to denote ``$S(h_t, a,x_t)$ for all $x_t \in X$ in the support''. This object is well defined in equilibrium for on-path histories and actions because the support of $x_t$ is type-independent and equilibrium beliefs must be consistent. We are adopting the simplifying assumption that the same holds off the equilibrium path, but this is not necessary for the arguments to go through as long as (NDOC-P) holds.}
\end{lemma}

\begin{proof}
	Assume by contradiction that $\bar{S}(h_t, a') < \underline{S}(h_t, a'')$ for some $a',a'' \in A(h_t)$. Pick any type $\theta \in S(h_t, a')$ and any strategy $\mathbf{a}'$ on path for $\theta$ at $h_t$. 
	Construct strategy $\mathbf{a}''$ as $\mathbf{a}''(h_t) = a''$ and $\mathbf{a}''(h_s) = \mathbf{a}'(h_s)$ for all $h_s \succ h_t$. This strategy constitutes a profitable deviation for $\theta$ at $h_t$. To see this, observe that the agent's lifetime utility can be written as
	\begin{align*}
		U(\mathbf{a} | h_t, \mathbf{b}, \theta) \equiv \mathbb{E} & \Bigg[ \tilde{u} \Big(\mathbf{a}(h_t), p(h_t,\mathbf{a}(h_t),x_t), \theta \Big) dt +  
		\\
		\nonumber
		&+ \sum_{s \in \mathcal{T}, s > t} e^{-r(s-t)} \tilde{u} \Big(\mathbf{a}(h_s), p(h_s,\mathbf{a}(h_s),x_s), \theta \Big) dt \mid h_t, \theta \Bigg].
	\end{align*}
	Since $p(h_s,\mathbf{a}''(h_s),x_s) \geq_{FOSD} \delta_{\underline{S} (h_t,a'')} >_{FOSD} \delta_{\bar{S} (h_t,a')} \geq_{FOSD} p(h_s,\mathbf{a}'(h_s),x_s)$ for any $x_s$ and all $h_s \succ h_t$, (MON) implies that
	\begin{align}
		\label{eq:Udiff}
		U(\mathbf{a}'' | h_t, \mathbf{b}, \theta) - U(\mathbf{a}' | h_t, \mathbf{b}, \theta) &\geq \mathbb{E} \Big[ \tilde{u} \big(a'', p(h_t,a'',x_t), \theta \big) - \tilde{u} \big(a', p(h_t,a',x_t), \theta \big) \Big] dt + \varDelta U,
	\end{align}
	where
	\begin{align*}
		\varDelta U \equiv \mathbb{E} \left[ \sum_{s \in \mathcal{T}, s > t} e^{-r(s-t)} \left[ \tilde{u} \Big(\mathbf{a}''(h_s), \delta_{\underline{S} (h_t,a'')}, \theta \Big) - \tilde{u} \Big(\mathbf{a}'(h_s), \delta_{\bar{S} (h_t,a')}, \theta \Big) \right] dt \mid h_t, \theta \right].
	\end{align*}
	By construction of $\mathbf{a}''$, $\mathbf{a}''(h_s) = \mathbf{a}'(h_s)$ for all $h_s \succ h_t$, hence $\tilde{u} \Big(\mathbf{a}''(h_s), \delta_{\underline{S} (h_t,a'')}, \theta \Big) = \tilde{u} \Big(\mathbf{a}'(h_s), \delta_{\underline{S} (h_t,a'')}, \theta \Big)$. By assumption, ${\underline{S} (h_t,a'')} > {\bar{S} (h_t,a')}$, hence (MON) implies $\tilde{u} \Big(\mathbf{a}''(h_s), \delta_{\underline{S} (h_t,a'')}, \theta \Big) > \tilde{u} \Big(\mathbf{a}'(h_s), \delta_{\bar{S} (h_t,a')}, \theta \Big)$ for all $h_s \succ h_t$. We conclude that $\varDelta U > 0$.
	Further, value of $\varDelta U$ is strictly positive regardless of $dt$.\footnote{This is in the sense that for any limit strategy $\mathbf{a}'$ and any outcome path $\mathbf{x}$, it holds by (MON) that $\int_t^\infty  e^{-r(s-t)} \left[ \tilde{u} \Big(\mathbf{a}'(h_s), \delta_{\underline{S} (h_t,a'')}, \theta \Big) - \tilde{u} \Big(\mathbf{a}'(h_s), \delta_{\bar{S} (h_t,a')}, \theta \Big) \right] dt > 0$, hence the expectation is strictly positive as well.}
	For $dt$ small enough, it dominates the first term on the RHS of \eqref{eq:Udiff} because $u$ is bounded, implying that for small enough $dt$, $U(\mathbf{a}'' | h_t, \mathbf{b}, \theta) > U(\mathbf{a}' | h_t, \mathbf{b}, \theta)$ which contradicts $\mathbf{a}'$ being optimal for $\theta$ at $h_t$.
	
	Now suppose $\bar{S}(h_t, a') = \underline{S}(h_t, a'')$. Suppose by way of contradiction that $|S(h_t, a')| > 1$, meaning $\underline{S}(h_t, a') < \underline{S}(h_t, a'')$.
	Let $p' \equiv \mathbb{E} [p(h_t,a',x_t) \mid h_t,a']$ be the belief induced by action $a'$ alone (note $p' <_{FOSD} \delta_{\bar{S} (h_t,a')} = \delta_{\underline{S}(h_t, a'')}$). 
	Consider type $\underline{\theta} \equiv \underline{S}(h_t,a')$. This type $\underline{\theta}$ must have an on-path strategy $\mathbf{a}'$ that yields reputation $p \not >_{FOSD} p'$ at all histories $h_s \succ (h_t,a')$ with positive probability (with certainty if outcomes $x$ are uninformative).\footnote{This follows from the observation that beliefs are correct in equilibrium, hence at any $h_s$ there must exist an action $a_s$ on path for $\underline{\theta}$ at $h_s$ that weakly increases the probability that the receiver assigns to type $\underline{\theta}$: $\mathbb{E}[p(h_s,a_s,x_s)(\underline{\theta}) \mid h_s, a_s] \geq p(h_s)(\underline{\theta})$.} 
	Construct strategy $\mathbf{a}''$ as above. Then inequality \eqref{eq:Udiff} holds with
	\begin{align*}
		\varDelta U \equiv \mathbb{E} \left[ \sum_{s \in \mathcal{T}, s > t} e^{-r(s-t)} \left[ \tilde{u} \Big(\mathbf{a}''(h_s), \delta_{\underline{S} (h_t,a'')}, \theta \Big) - \tilde{u} \Big(\mathbf{a}'(h_s), p(h_s,\mathbf{a}'(h_s),x_s), \theta \Big) \right] dt \mid h_t, \theta \right] > 0,
	\end{align*} 
	which is positive since $\tilde{u} \Big(\mathbf{a}''(h_s), \delta_{\underline{S} (h_t,a'')}, \theta \Big) = \tilde{u} \Big(\mathbf{a}'(h_s), \delta_{\underline{S} (h_t,a'')}, \theta \Big) > \tilde{u} \Big(\mathbf{a}'(h_s), p(h_s,\mathbf{a}'(h_s),x_s), \theta \Big)$ for all $h_s \succ h_t$ by the same argument as above.
	Unlike in the previous case, $\varDelta U$ is not bounded away from zero for all $dt$. However, for the IC condition $U(\mathbf{a}' | h_t, \mathbf{b}, \theta) \geq U(\mathbf{a}'' | h_t, \mathbf{b}, \theta)$ to hold, we must have that $\varDelta U \to 0$ as $dt \to 0$. This would require that $p' \to \delta_{\bar{S} (h_t,a')}$ as $dt \to 0$, which yields a contradiction in the limit, since $S(p') = S(h_t,a') \supsetneq \{ \bar{S}(h_t,a') \}$.
\end{proof}

Next, Lemma \ref{lem:onion1} puts the (SC) property to use, establishing a form of monotonicity of optimal strategies w.r.t. type (``higher types play higher strategies''). The main problem in the dynamic setting is the lack of any nice complete order over strategies $\mathbf{a}$, so given two arbitrary strategies, we generally cannot say which one of them is ``higher''. Therefore, we rephrase this monotonicity result to say instead that if a given strategy (or its equivalent) is optimal for two agent types, then it must also be optimal for all types in between. We cannot say with certainty that the given strategy is chosen on equilibrium path by any of these types in between, but we can claim that any strategy they play must be payoff-equivalent to the one under consideration.

\begin{lemma} \label{lem:onion1}
	Suppose (SC) holds. Fix any equilibrium $(\alpha,\mathbf{b},\theta)$ (for any $dt>0$) and history $h_t \in \mathcal{H}$.
	If there exists a pair of strategies $\underline{\mathbf{a}},\bar{\mathbf{a}} \curlywedge h_t$ that are payoff-equivalent at $h_t$ and are on path at $h_t$ for some types $\underline{\theta}$ and	$\bar{\theta} > \underline{\theta}$ respectively, then any strategy $\hat{\mathbf{a}} \curlywedge h_t$ on path at $h_t$ for any $\hat{\theta} \in (\underline{\theta}, \bar{\theta})$ must be payoff-equivalent at $h_t$ to $\bar{\mathbf{a}},\underline{\mathbf{a}}$.
\end{lemma}

\begin{proof}
	Fix any such $\hat{ \mathbf{a}}$.
	Strategy $\bar{\mathbf{a}}$ has to be optimal for type $\bar{\theta}$. In particular, when evaluated at $h_t$, it has to be better than $\hat{\mathbf{a}}$:
	\begin{align*}
	U (\bar{\mathbf{a}} | h_t, \mathbf{b}, \bar{\theta}) \geq U (\hat{\mathbf{a}} | h_t, \mathbf{b}, \bar{\theta}).
	\end{align*}
	The same holds for type $\underline{\theta}$, since $\bar{\mathbf{a}}$ and $\underline{\mathbf{a}}$ are payoff-equivalent:
	\begin{align*}
	U (\bar{\mathbf{a}} | h_t, \mathbf{b}, \underline{\theta}) = U (\underline{\mathbf{a}} | h_t, \mathbf{b}, \underline{\theta}) \geq U (\hat{\mathbf{a}} | h_t, \mathbf{b}, \underline{\theta}).
	\end{align*}
	At the same time, $\hat{\theta}$ at least weakly prefers $\hat{\mathbf{a}}$ to $\bar{\mathbf{a}}$, meaning that the converse holds for $\hat{\theta}$:
	\begin{equation*}
	U (\bar{\mathbf{a}} | h_t, \mathbf{b}, \hat{\theta}) \leq U (\hat{\mathbf{a}} | h_t, \mathbf{b}, \hat{\theta}).
	\end{equation*}
	If this inequality is strict, then this is a direct contradiction with (SC), which requires that $U (\bar{\mathbf{a}} | h_t, \mathbf{b}, \theta) - U (\hat{\mathbf{a}} | h_t, \mathbf{b}, \theta)$ as a function of $\theta$ either crosses zero at most once, or is exactly zero.
\end{proof}

Lemma \ref{lem:onion2} below is the final step before we can move on to the proofs of main results. It can be seen as a weaker version of Proposition \ref{prop:attrn}, claiming that the highest and lowest types at any history have a strategy in common.

\begin{lemma} \label{lem:onion2}
	Suppose (MON) and (SC) hold and $dt \to 0$. Fix an equilibrium $(\alpha,\mathbf{b},p)$ and any $h_t \in \mathcal{H}$. 
	There exist $h_t$-payoff-equivalent strategies $\bar{\mathbf{a}}, \underline{\mathbf{a}} \curlywedge h_t$, on path at $h_t$ for $\bar{S}(h_t)$ and $\underline{S}(h_t)$ respectively.
\end{lemma}

\begin{proof}
	We will proceed by induction on the support size $|S(h_t)|$. The claim of the lemma holds trivially for $|S(h_t)|=1$, and by Theorem \ref{thm:2types} it also holds for $|S(h_t)|=2$. The remainder of the proof shows that if the claim holds when $|S(h_t)|=k-1$ for $k \geq 3$, then it also holds when $|S(h_t)|=k$.
	Let $\mathbf{x} : \mathcal{H} \to X$ denote an outcome profile which prescribes some outcome for every history. Fix some $\mathbf{x}$. Coupled with some pure strategy of the agent, the receiver's equilibrium strategy $\mathbf{b}$ and the equilibrium belief system $p$, it fully determines the path of play and the agent's payoffs. 
	
	Begin the second layer of induction, iterating forwards on time periods from $t$. At $h_t$ and any subsequent history $h_s \succ h_t$, one of the following must apply:
	\begin{enumerate}
		\item There is an action $a$ on path for both types $\bar{S}(h_t)$ and $\underline{S}(h_t)$ at $h_s$. If this is the case, call $h_s$ a \emph{non-splitting} history and continue to $h_{s+dt} = (h_s, a, \mathbf{x}(h_s))$.
		\item There is no action $a$ on path for both $\bar{S}(h_t)$ and $\underline{S}(h_t)$ at $h_s$. If this is the case, call $h_s$ a \emph{splitting} history.
	\end{enumerate}
	Proceed along the non-splitting path (according to the chosen $\mathbf{x}$) until the first splitting history $h_s$. Pick arbitrary actions $\bar{a}$ and $\underline{a}$ that are on path for $\bar{S}(h_t)$ and $\underline{S}(h_t)$ at $h_s$ respectively, and consider two continuation histories $\bar{h}_{s+dt} \equiv (h_s, \bar{a}, \mathbf{x}(h_s))$ and $\underline{h}_{s+dt} \equiv (h_s, \underline{a}, \mathbf{x}(h_s))$. Then we have that $|S(h_{s+dt})| < |S(h_s)| = k$ for both continuation histories, because $S(\bar{h}_{s+dt}) \subseteq S(h_s) \backslash \underline{S}(h_s)$ and $S(\underline{h}_{s+dt}) \subseteq S(h_s) \backslash \bar{S}(h_s)$. Therefore, by the induction assumption, the statement of the lemma holds at both $\bar{h}_{s+dt}$ and $\underline{h}_{s+dt}$. 
	
	In particular, the statement of the lemma for $\bar{h}_{s+dt}$ states that there exist two $\bar{h}_{s+dt}$-payoff-equivalent strategies on path at $\bar{h}_{s+dt}$ for $\bar{S}(h_s)$ and $\underline{S}(\bar{h}_{s+dt})$ respectively.
	Playing $\bar{a}$ at $h_s$ is on path for both of these types, hence there also exists a pair of strategies on path at $h_s$ for the two types respectively, which grant the same payoff \emph{conditional on $\mathbf{x}$}.\footnote{It does not matter for our argument if all types assign probability zero to outcome $\mathbf{x}(h_s)$ conditional on $\bar{a}$.}
	However, the argument above applies to any outcome profile $\mathbf{x}$ and, in particular, to any outcome $\mathbf{x}(h_s)$, hence there also exists a pair of strategies $\bar{ \mathbf{a}}', \bar{ \mathbf{a}}''$ on path at $h_s$ for the two types $\bar{S}(h_s)$ and $\underline{S}(\bar{h}_{s+dt})$ respectively, which are payoff-equivalent at $h_s$ (unconditionally).
	
	By a mirror argument, there also exists a pair of $h_s$-payoff-equivalent strategies $\underline{ \mathbf{a}}', \underline{ \mathbf{a}}''$ on path at $h_s$ for $\underline{S}(h_s)$ and $\bar{S}(\underline{h}_{s+dt})$ respectively. Note further that by Lemma \ref{lem:supp1} we have that $\bar{S}(\underline{h}_{s+dt}) > \underline{S}(\bar{h}_{s+dt})$. Lemma \ref{lem:onion1} hence applies: $\underline{ \mathbf{a}}''$ must be payoff-equivalent to $\bar{ \mathbf{a}}', \bar{ \mathbf{a}}''$, thus so is $\underline{ \mathbf{a}}'$. We have shown that the statement of the lemma holds at $h_t$ if $|S(h_t)|=k$ and $h_t$ is a splitting history.
	
	We are left to cover non-splitting histories. Suppose $h_t$ is non-splitting. Fix $\mathbf{x}$. Then we know that the statement of the lemma holds at the first splitting history $h_s$ following $h_t$ along the path of pooling actions and fixed outcomes $\mathbf{x}$. Therefore, there exists a pair of strategies on path at $h_t$ for $\bar{S}(h_t)$ and $\underline{S}(h_t)$, which grant the same payoff at $h_t$ conditional on $\mathbf{x}$. This applies to any outcome profile $\mathbf{x}$, hence there exists a pair of strategies $\bar{\mathbf{a}}, \underline{\mathbf{a}}$ on path at $h_t$ for $\bar{S}(h_t)$ and $\underline{S}(h_t)$, which are payoff-equivalent at $h_t$. This concludes the induction argument and the proof of the lemma.
\end{proof}

\begin{proof}[Proof of Proposition \ref{prop:attrn}.]
	Let $\bar{\theta} \equiv \bar{S}(h_t)$. Note that the statement of the proposition holds trivially if $|S(h_t)|=1$, so for the remainder of this proof we assume that this is not the case (i.e., $\bar{ \theta} \neq \underline{ \theta}$).
	From Lemma \ref{lem:onion2} we know there exist $h_t$-payoff-equivalent $\bar{ \mathbf{a}}, \underline{ \mathbf{a}} \curlywedge h_t$ on path at $h_t$ for $\bar{\theta}$ and $\underline{\theta}$ respectively. Then by Lemma \ref{lem:onion1}, any pure strategy $\mathbf{a} \curlywedge h_t$ on path at $h_t$ for any $\theta \in S(h_t) \backslash \{ \bar{\theta}, \underline{\theta} \}$ is payoff-equivalent at $h_t$ to $\bar{ \mathbf{a}}, \underline{ \mathbf{a}}$.
	
	Suppose now there exists a pure strategy $\bar{ \mathbf{a}}' \curlywedge h_t$ on path at $h_t$ for $\bar{\theta}$, which is payoff-distinct at $h_t$ from $\bar{ \mathbf{a}}$. By (SC), all types $\theta \in S(h_t) \backslash \bar{\theta}$ must have a strict preference at $h_t$ between $\bar{ \mathbf{a}}$ and $\bar{ \mathbf{a}}'$. The former is optimal for these types, hence $\bar{ \mathbf{a}}'$ is only on path for $\bar{\theta}$. The two strategies cannot prescribe different actions $\bar{ \mathbf{a}}(h_t) \neq \bar{ \mathbf{a}}'(h_t)$ at $h_t$, since this is in violation of Lemma \ref{lem:supp1}. The same, however, applies to any subsequent history $h_s$ such that $|S(h_s)|>1$. At all $h_s \succ h_t$ s.t. $|S(h_s)|=1$, Lemma \ref{lem:obs1} implies that all pure strategies on path at $h_s$ are $h_s$-payoff-equivalent. Both facts together imply that $\bar{ \mathbf{a}}(h_s) = \bar{ \mathbf{a}}'(h_s)$ for all $h_s \succ h_t$. This contradicts $\bar{ \mathbf{a}}$ and $\bar{ \mathbf{a}}'$ being payoff-distinct, hence such $\bar{ \mathbf{a}}'$ does not exist.
	Therefore, any pure strategy $\mathbf{a}$ on path at $h_t$ that is $h_t$-payoff-distinct from $\bar{ \mathbf{a}}$ is only on path for $\underline{ \theta}$. This concludes the proof.
\end{proof}

\begin{proof}[Proof of Theorem \ref{thm:main}.]
	From Proposition \ref{prop:attrn}, all actions $a \in A(h_t)$ are on path for $\underline{ \theta}$, which proves the first statement of the theorem.
	
	From the fact that payoff-relevant signaling happens at $h_t$ we know that there exist two pure strategies $\underline{ \mathbf{a}},\bar{ \mathbf{a}}$ that are payoff-distinct at $h_t$ and prescribe different actions at $h_t$: $\underline{a} \equiv \underline{ \mathbf{a}}_t \neq \bar{a} \equiv \bar{ \mathbf{a}}_t$. From Proposition \ref{prop:attrn} we know at least one of these strategies -- suppose $\underline{ \mathbf{a}}$ -- is on path for $\underline{ \theta}$ but not for any other $\theta \in S(h_t) \backslash \bar{ \theta}$ at $h_t$.
	Furthermore, it follows from the definition of payoff-relevant signaling that there is no $\bar{ \mathbf{a}}'$ s.t. $\bar{ \mathbf{a}}' \curlywedge h_t$ and $\bar{ \mathbf{a}}'_t = \underline{a}$, and which is payoff-equivalent to $\bar{ \mathbf{a}}$ at $h_t$. Therefore, $\underline{a}$ is only on path for $\underline{ \theta}$, while $\bar{a}$ is optimal for all $\theta \in S(h_t)$ at $h_t$.
	 
	We now show that $\underline{a} \in A^*\left( h_t, \mathbf{b}, \underline{\theta} \right)$. Suppose not. Then type $\underline{\theta}$ can play some $a_s \in A^*\left( h_s, \mathbf{b}, \underline{\theta} \right)$ at every history $h_s \succeq h_t$. Compared to following $\underline{ \mathbf{a}}$, this strategy would yield the same payoff at all times $s > t$ and a strictly higher payoff at $t$ (same as in the proof of Theorem \ref{thm:2types}), hence $\underline{ \mathbf{a}}$ is not optimal for $\underline{\theta}$ at $h_t$ -- a contradiction.
	
	To complete the proof of statements 2 and 3 of the theorem, we need to show that $\bar{a} \notin A^*\left( h_t, \mathbf{b}, \underline{\theta} \right)$. Assume not. Consider the strategy of playing $\bar{a}$ at $h_t$ and all subsequent histories. Compared to following $\underline{ \mathbf{a}}$, this strategy would yield $\underline{\theta}$ a weakly higher payoff at all times $s > t$ and a strictly higher payoff at $t$ (due to $p(h_t, (\bar{a},x_t)) > \delta_{\underline{\theta}}$ for all $x_t$ and to the strict part of (MON)), hence it is a profitable deviation from $\underline{ \mathbf{a}}$ for $\underline{\theta}$ at $h_t$ -- a contradiction.
	
	This completes the proof of Theorem \ref{thm:main}.
\end{proof}
	
\begin{proof}[Proof of Corollary \ref{cor:repsig}.]
	The low type must be indifferent between taking a separating action $\underline{a}$ at $h_t$ and pooling on $\bar{a}$ at $h_t$ and separating at $h_{t+dt}$. This indifference dictates that one period of pooling must be exactly as attractive as one period of being revealed as $\underline{ \theta}$, i.e., $\mathbb{E}_x \left[ \tilde{u} \left(\bar{a},p(h_t, (\bar{a},x)), \underline{\theta}\right) | \underline{ \theta} \right] = \tilde{u} \left(\underline{a}, \delta_{\underline{\theta}}, \underline{\theta} \right)$.
\end{proof}

\begin{proof}[Proof of Corollary \ref{cor:supp}.]
	Proposition \ref{prop:attrn} states that all pure strategies $\mathbf{a}'$ on path at $h_t$ for any $\theta \in S(h_t) \backslash \underline{ S}(h_t)$ are payoff-equivalent at $h_t$. Since there is no payoff-relevant signaling in equilibrium, the set of such strategies is a singleton: if there is more than one then there exists $h_s \succ h_t$ at which the two prescribe different actions, but that constitutes payoff-irrelevant signaling at $h_s$ (the two strategies coincide on $[t,s)$, hence they are payoff-equivalent at $h_s$).
	
	Therefore, at any $h_t$ there exists some $\bar{a} \in A$ such that $\alpha_\theta(\bar{a}|h_t)=1$ for all $\theta \in S(h_t) \backslash \underline{ S}(h_t)$. Together with part 1 of Theorem \ref{thm:main}, this means that $S(h_t, \bar{a}) = S(h_t)$. By part 3 of the theorem, $\bar{a}$ is the unique element of $A(h_t) \backslash A^*\left(  h_t, \mathbf{b}, \underline{\theta} \right)$. By part 2 of the theorem, for any $\underline{a} \in A(h_t) \cap A^*\left( h_t, \mathbf{b}, \underline{\theta} \right)$ we have $S(h_t, \underline{a}) = \underline{S}(h_t)$. Since all on-path histories $h_{t+dt}$ can be written as $h_{t+dt} = (h_t, a, x_t)$ for some $a \in A(h_t)$ and $x_t \in X$, and outcome $x_t$ does not change support $S$, we obtain that for any pair of on-path histories $h_t,h_{t+dt}$, it must be that either $S(h_{t+dt})=S(h_t)$, or  $S(h_{t+dt})= \{\underline{S}(h_t)\}$. Applying this observation iteratively from $h_0$ (for which $S(h_0) = \Theta$) completes the proof.
\end{proof}

\subsection{Proofs: Separable Settings}

\begin{proof}[Proof of Proposition \ref{prop:suff_mon}.]
	It is immediate that if $\phi_1(a,p)$ is weakly increasing in $p$ and $\psi(\theta) \geq 0$ then $\phi_1(a,p) \psi(\theta)$ is weakly increasing in $p$. Together with the assumption that $\phi_0(a,p)$ is weakly increasing in $p$, this implies trivially that \eqref{eq:u_repr} is weakly increasing in $p$, which is what is required by (MON).
\end{proof}

\begin{proof}[Proof of Proposition \ref{prop:suff_sc}.]
	For some fixed $\mathbf{a}',\mathbf{a}''$, $h_t$, the function $\mathcal{U}(\theta)$ under representation \eqref{eq:u_repr} is given by
	\begin{align*}
		\mathcal{U}(\theta) 
		&= \mathbb{E} \left[ \sum_{s \in \mathcal{T}, s \geq t} \left[ \phi_0 \left( \mathbf{a}''(h_s),p(h_s) \right) - \phi_0 \left(\mathbf{a}'(h_s),p(h_s) \right) \right] dt + \right. 
		\\
		& \hphantom{= \mathbb{E} \Bigg[ }
		\left. + \sum_{s \in \mathcal{T}, s \geq t} \left[ \phi_1 \left( \mathbf{a}''(h_s),p(h_s) \right) - \phi_1 \left(\mathbf{a}'(h_s),p(h_s) \right) \right] \psi(\theta) dt \mid h_t, \mathbf{b}, \theta \right] 
		\\
		&= \left( \sum_{s \in \mathcal{T}, s \geq t} \left[ \phi_1 \left( \mathbf{a}''(h_s),p(h_s) \right) - \phi_1 \left(\mathbf{a}'(h_s),p(h_s) \right) \right] dt \right) \psi(\theta),
	\end{align*}
	where the expectation in the first two lines is over outcomes, hence by assumption that outcomes are uninformative, the expectation is resolved trivially, since strategies $\mathbf{a}, \mathbf{b}$ then fully determine the path of play.
	It then follows immediately that if $\psi(\theta)$ is strictly monotone, then $\mathcal{U}(\theta)$ is either strictly monotone, or equivalently zero (if the term in brackets evaluates to zero), hence (SC) holds.
\end{proof}

\section{Informative Equilibrium in the Price Signaling Model} \label{sub:antitrust_exp}

This appendix constructs an informative equilibrium in the setting of Section \ref{sub:antitrust}.
Let $\Theta = \{1,2,3\}$, $\lambda(a) = \lambda a$ for some constant $\lambda >0$, and $\gamma(p) = 1 - \mathbb{E}[\theta|p]/3$. We will construct an equilibrium in which payoff-relevant signaling occurs in all periods $t \in [0,\bar{t}] \cap \mathcal{T}$ for some $\bar{t} \in \mathcal{T}$. According to Theorem \ref{thm:main}, in an equilibrium with no payoff-irrelevant signaling, there will be two prices on the equilibrium path: one pooling price $a_P(h_t)$ set by all types with positive probability and a separating price $a_S(h_t)=a_S$ only set with positive probability by the lowest type $\theta=1$.\footnote{Note that the lowest type $\theta=1$ does, in fact, reap the highest payoff of any type. However, $\theta=1$ is the least desirable type reputation-wise, since it increases the probability of being regulated by the authority.}
The separating price must be a bliss price for the lowest type $\theta=1$. The bliss price for type $\theta$ can be calculated by maximizing the flow payoff:
\begin{align*}
	a_S(\theta) &= \arg \max_{a \in A} \left\{ \tilde{u}(a,p,\theta) \right\}
	\\
	&= \arg \max_{a \in A} \left\{ a (1 - \theta a) - \lambda a \gamma(p) F \right\}
	\\
	&= \arg \max_{a \in A} \left\{ a z(p_t) - \theta a^2 \right\}
	\\
	&= \frac{z(p_t)}{2\theta},
\end{align*}
where $z(p_t) \equiv 1 - \lambda F \left(1 - \frac{\mathbb{E}[\theta|p_t]}{3} \right)$. Separation (or, by (NDOC-P), any other deviation from the pooling strategy) perfectly reveals the lowest type, hence $p_t = \delta_1 \equiv \delta_{\theta=1}$ in that case, and so the separating price for type $\theta=1$ is $a_S(1) = \frac{z(\delta_{1})}{2}$.

To calculate the pooling price $a_P(h_t)$, invoke Corollary \ref{cor:repsig} which states the the lowest type $\theta=1$ must receive the same payoff from both pooling and separating:
\begin{align*}
	\tilde{u}(a_P(p_t),p_t,1) &= \tilde{u}(a_S(1),\delta_1,1)
	\\
	\Leftrightarrow a_P(p_t) z(p_t) - a_P^2(p_t) &= a_S(1) z(\delta_1) - a_S^2(1)
	\\
	\Rightarrow a_P(p_t) &= \frac{1}{2} \left[ z(p_t) \pm \sqrt{z^2(p_t) - z^2(\delta_1)} \right].
\end{align*}
It makes sense to select the negative root, since the intuition behind the problem suggests $a_P(h_t) < a_S(1)$: convincing the regulator that the market is competitive involves setting prices below the monopoly price, not above. 

To finish the equilibrium construction, it remains to pick an arbitrary prior belief $p_0$, per-period separation probability for the lowest type (which pins down the evolution of $p_t$), calculate the resulting path of $a_P(h_t)$, and verify that it is incentive compatible for the competitive types to join in on the pooling price $a_P(h_t)$, rather than separating away to $a_S(\theta) = \frac{z(\delta_1)}{2 \theta}$. This IC condition is difficult to verify analytically (and it may or may not hold, depending on the parameters, possibly leading to nonexistence of an informative equilibrium), but can be verified numerically. 

\begin{figure}[h]
	\centering
	\subfloat[The pooling and the pre-investigation bliss prices.]{
		\includegraphics[scale=0.5]{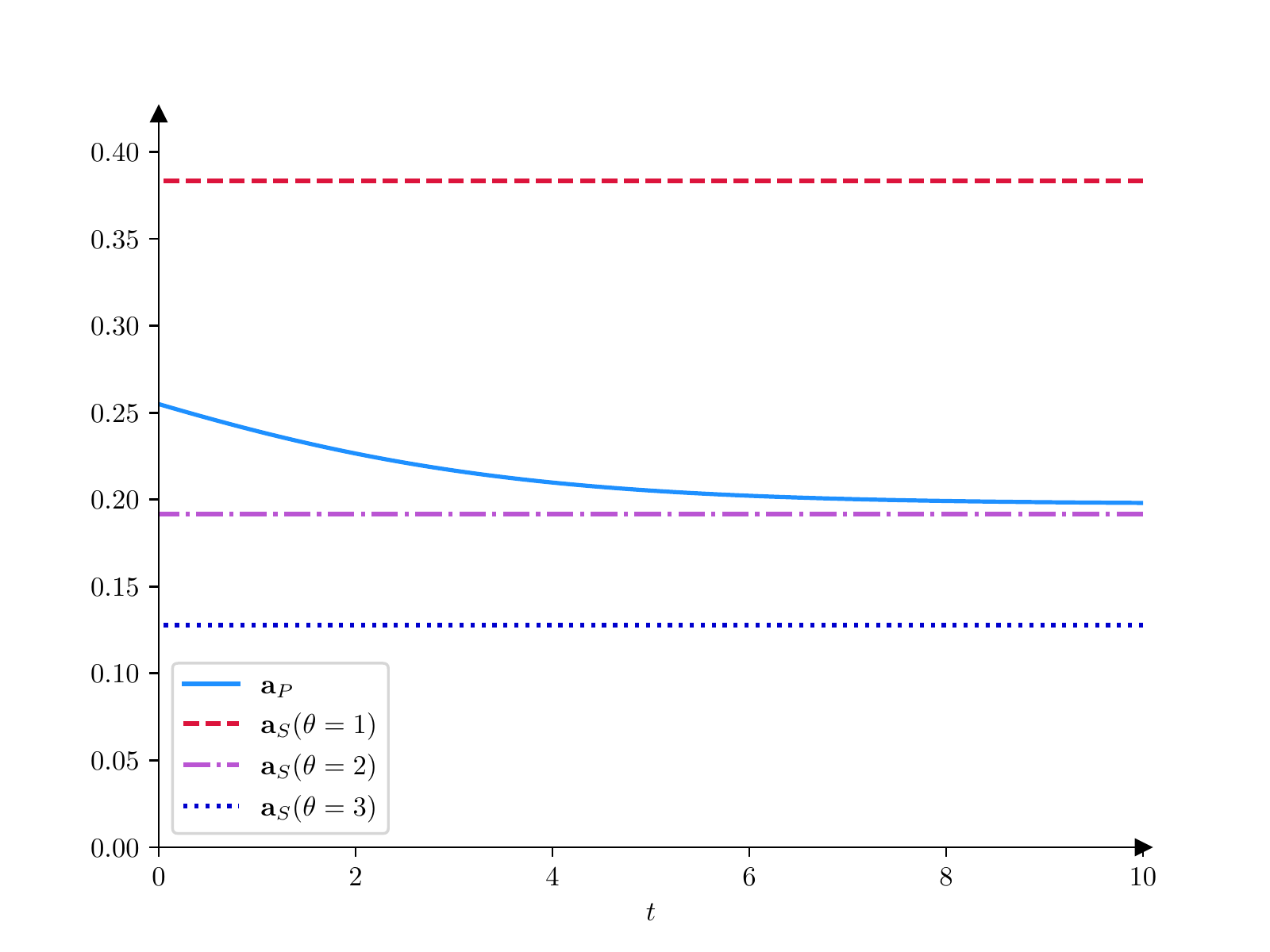}
		\label{fig:eq:1}
	}
	\subfloat[Values of the three types from pooling and separating.]{ 
		\includegraphics[scale=0.5]{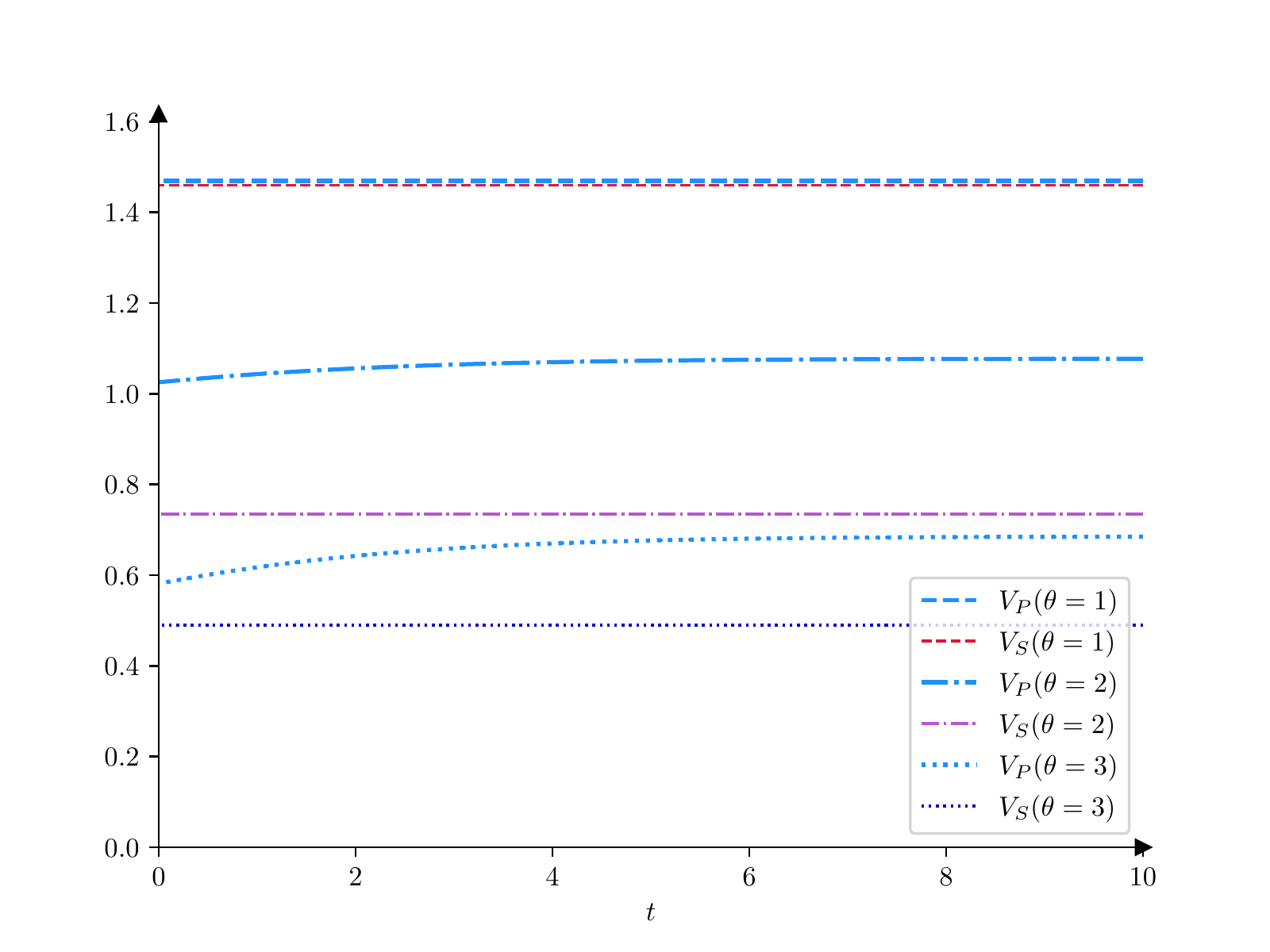}
		\label{fig:eq:2}
	}
	\\
	\subfloat[The evolution of the regulator's beliefs.]{
		\includegraphics[scale=0.5]{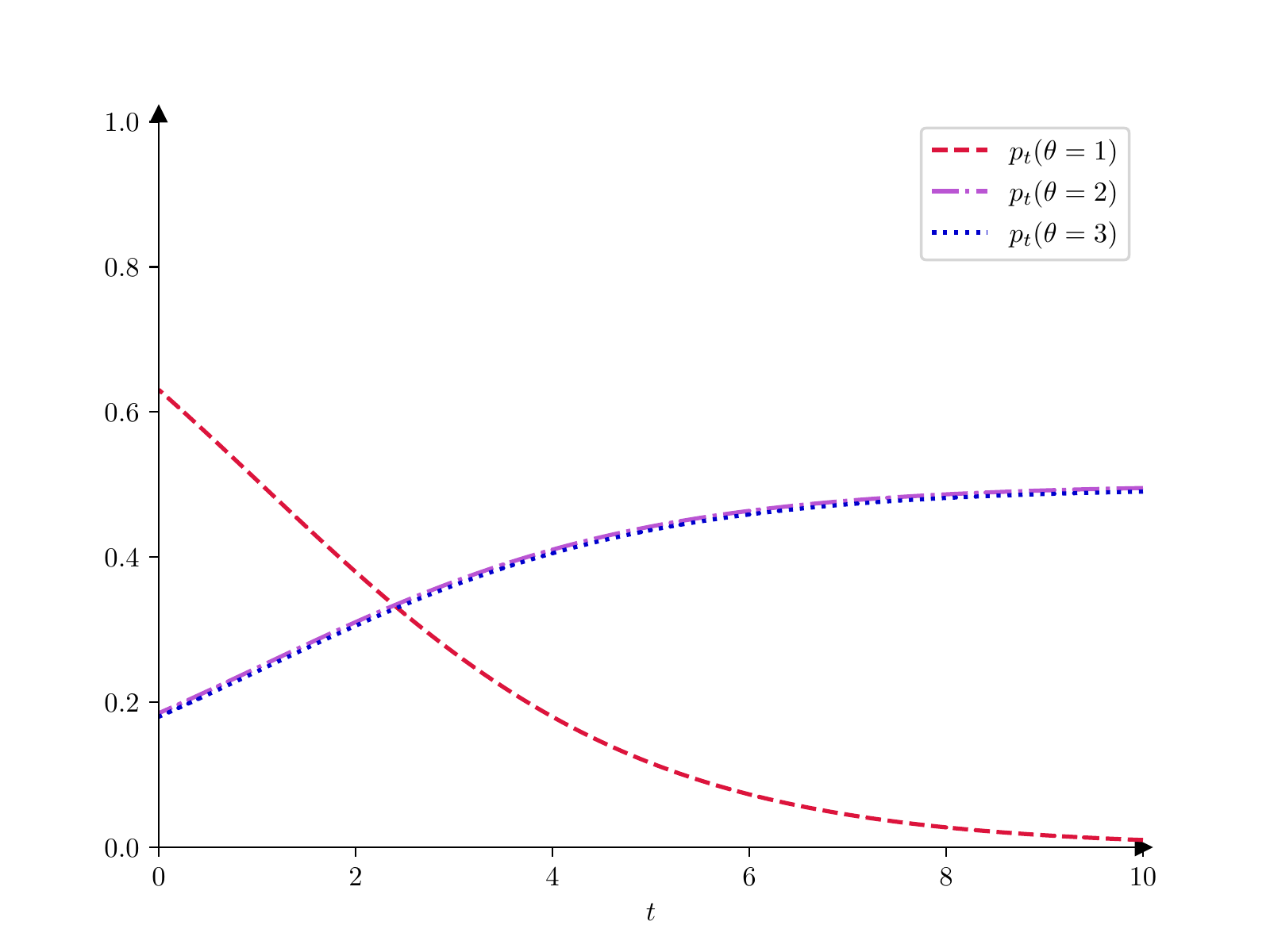}
		\label{fig:eq:3}
	}
	\subfloat[Firm's type, as expected by the regulator.]{ 
		\includegraphics[scale=0.5]{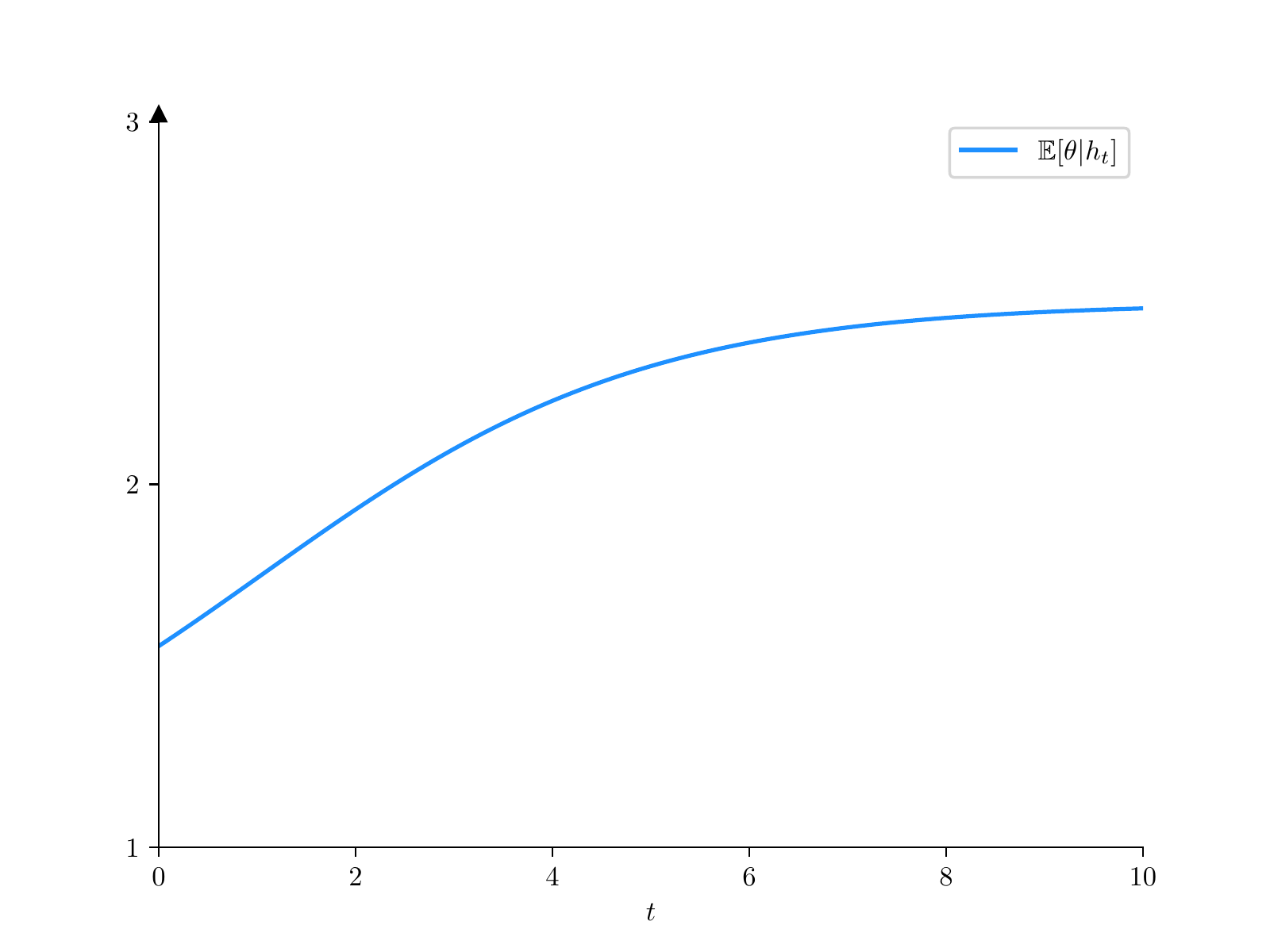}
		\label{fig:eq:4}
	}
	\caption{Example attrition equilibrium ($r=0.1$, $\lambda=0.35$, $F = 1$, separation intensity $0.5$). }
	\label{fig:eq}
\end{figure}

Figure \ref{fig:eq} presents a numerical example with period length $dt=0.1$. In this example, separation occurs over $t \in [0,10] \cap \mathcal{T}$, and type $\theta=1$ separates with intensity $0.5$ (which is equal to separation probability $0.5dt = 0.05$ per period) over that time interval. From $t=10$ onwards, all types switch to pooling indefinitely. Instead of $p_0$ we actually fix $p_{10} = (0.01,0.495,0.495)$ and unravel to arrive at $p_0 = (0.63, 0.185, 0.185)$. For the remaining parameters we assume $r=0.1$, $\lambda = 0.35$, $F = 1$.

From Panel \ref{fig:eq:1} we can see that the pooling price in this example is higher than the competitive prices, i.e., the bliss price that types $\theta=2$ or $\theta=3$ would set, but this observation is specific to the parameters. The latter follows from Corollary \ref{cor:repsig}, which mentions that the pooling action is dictated by the lowest type's indifference and is completely independent from other types' preferences. Panel \ref{fig:eq:2} demonstrates visually that all the IC constraints hold: the monopolist $\theta=1$ is indifferent between pooling and separating at all $t$, while the competitive types $\theta\in \{2,3\}$ always prefer pooling. Panels \ref{fig:eq:3} and \ref{fig:eq:4} demonstrate that the extent of information revelation can still be substantial in equilibrium: the regulator becomes more convinced of market competitiveness over time, with the probability she assigns to $\theta=1$ dropping from $0.63$ to $0.01$ along the pooling path.

\end{document}